\documentclass[11pt]{article}

\usepackage{amsthm}
\usepackage{graphicx} 
\usepackage{array} 

\usepackage{amsmath, amssymb, amsfonts, verbatim}
\usepackage{hyphenat,epsfig,subfigure,multirow}

\usepackage[usenames,dvipsnames]{xcolor}
\usepackage[ruled]{algorithm2e}

\usepackage{tcolorbox}
\tcbuselibrary{skins,breakable}
\tcbset{enhanced jigsaw}

\usepackage[compact]{titlesec}

\definecolor{DarkRed}{rgb}{0.5,0.1,0.1}
\definecolor{DarkBlue}{rgb}{0.1,0.1,0.5}

\usepackage{nameref}
\definecolor{ForestGreen}{rgb}{0.1333,0.5451,0.1333}
\definecolor{Red}{rgb}{0.9,0,0}
\usepackage[linktocpage=true,
	pagebackref=true,colorlinks,
	linkcolor=DarkRed,citecolor=ForestGreen,
	bookmarks,bookmarksopen,bookmarksnumbered]
	{hyperref}

\usepackage{bm}
\usepackage{url}
\usepackage{xspace}
\usepackage[mathscr]{euscript}

\usepackage{mdframed}

\usepackage[noend]{algpseudocode}
\makeatletter
\def\BState{\State\hskip-\ALG@thistlm}
\makeatother

\usepackage{cite}
\usepackage{enumitem}

\usepackage[margin=1in]{geometry}

\newtheorem{theorem}{Theorem}
\newtheorem{lemma}{Lemma}[section]
\newtheorem{proposition}[lemma]{Proposition}

\newtheorem{claim}[lemma]{Claim}

\newtheorem{definition}{Definition}

\newtheorem{problem}{Problem}
\newtheorem{remark}[lemma]{Remark}

\newtheorem*{claim*}{Claim}
\newtheorem*{proposition*}{Proposition}
\newtheorem*{lemma*}{Lemma}
\newtheorem*{problem*}{Problem}

\newtheorem{mdresult}{Result}
\newenvironment{Theorem}{\begin{mdframed}[backgroundcolor=lightgray!40,topline=false,rightline=false,leftline=false,bottomline=false,innertopmargin=2pt]\begin{mdresult}}{\end{mdresult}\end{mdframed}}

\newtheorem{result}[mdresult]{Result}

\newtheorem{mdinvariant}{Invariant}

\allowdisplaybreaks

\renewcommand{\qed}{\nobreak \ifvmode \relax \else
      \ifdim\lastskip<1.5em \hskip-\lastskip
      \hskip1.5em plus0em minus0.5em \fi \nobreak
      \vrule height0.75em width0.5em depth0.25em\fi}

\newcommand*\samethanks[1][\value{footnote}]{\footnotemark[#1]}

\newcommand{\toShrink}{-.20cm}
\newcommand{\toShrinkEnu}{-.2cm}

\newcommand{\jstar}{\ensuremath{{j^{\star}}}}

\newcommand{\Ot}{\ensuremath{\widetilde{O}}}
\newcommand{\eps}{\ensuremath{\varepsilon}}
\newcommand{\Paren}[1]{\Big(#1\Big)}
\newcommand{\Bracket}[1]{\Big[#1\Big]}
\newcommand{\bracket}[1]{\left[#1\right]}
\newcommand{\paren}[1]{\ensuremath{\left(#1\right)}\xspace}
\newcommand{\card}[1]{\left\vert{#1}\right\vert}

\newcommand{\norm}[1]{\ensuremath{\|#1\|}}

\newcommand{\prob}[1]{\Pr\paren{#1}}
\newcommand{\expect}[1]{\Exp\bracket{#1}}
\newcommand{\var}[1]{\textnormal{Var}\bracket{#1}}
\newcommand{\cov}[1]{\textnormal{Cov}\bracket{#1}}

\newcommand{\set}[1]{\ensuremath{\left\{ #1 \right\}}}
\newcommand{\poly}{\mbox{\rm poly}}
\newcommand{\polylog}{\mbox{\rm  polylog}}

\newcommand{\alg}{\ensuremath{\mathcal{A}}\xspace}

\DeclareMathOperator*{\Exp}{\ensuremath{{\mathbb{E}}}}
\DeclareMathOperator*{\Prob}{\ensuremath{\textnormal{Pr}}}
\renewcommand{\Pr}{\Prob}

\newcommand{\Ex}{\Exp}

\newcommand{\etal}{{\it et al.\,}}

\newenvironment{tbox}{\begin{tcolorbox}[
		enlarge top by=5pt,
		enlarge bottom by=5pt,
		 breakable,
		 boxsep=0pt,
                  left=4pt,
                  right=4pt,
                  top=10pt,
                  arc=0pt,
                  boxrule=1pt,toprule=1pt,
                  colback=white
                  ]
	}
{\end{tcolorbox}}

\newcommand{\Yes}{\ensuremath{\textnormal{\textsf{Yes}}}\xspace}
\newcommand{\No}{\ensuremath{\textnormal{\textsf{No}}}\xspace}

\newcommand{\prot}{\ensuremath{\pi}}

\newcommand{\bC}{\ensuremath{\bm{C}}}

\newcommand{\CC}[3]{\ensuremath{\textnormal{\textsf{CC}}_{#2}^{#3}(#1)}\xspace}

\title{Sublinear Algorithms for $(\Delta+1)$ Vertex Coloring}
\author{Sepehr Assadi\thanks{Department of Computer and Information Science, University of Pennsylvania. Supported in part by the National Science Foundation grant CCF-1617851. Email: {{\small {\tt \{sassadi,chenyu2,sanjeev\}@cis.upenn.edu.}}}}
\and Yu Chen\samethanks
\and Sanjeev Khanna\samethanks
}

\date{}

\begin{document}
\maketitle

\thispagestyle{empty}
\begin{abstract}
	
	Any graph with maximum degree $\Delta$ admits a proper vertex coloring with $\Delta+1$ colors that can be found via a simple sequential greedy algorithm in linear time and space. But can one find such a coloring 
	via a \emph{sublinear} algorithm? 

	\smallskip

	We answer this fundamental question in the affirmative for several canonical classes of sublinear algorithms including graph streaming, sublinear time, and massively parallel computation (MPC) algorithms. In particular, 
	we design:

	\begin{itemize}
		\item A \emph{single-pass} semi-streaming algorithm in dynamic streams using $\Ot(n)$ space. The only known semi-streaming algorithm prior to our work was a folklore $O(\log{n})$-pass algorithm obtained
		by simulating classical distributed algorithms in the streaming model. 
		
		\item A sublinear-time algorithm in the standard query model that allows neighbor queries and pair queries using $\Ot(n\sqrt{n})$ time. We further show that any algorithm that outputs a valid coloring with
		sufficiently large constant probability requires $\Omega(n\sqrt{n})$ time. No non-trivial sublinear time algorithms were known prior to our work. 
		
		\item A parallel algorithm in the massively parallel computation (MPC) model using $\Ot(n)$ memory per machine and $O(1)$ MPC rounds. Our number of rounds significantly improves upon the recent $O(\log\log{\Delta}\cdot\log^*{(n)})$-round
		algorithm of Parter [ICALP 2018]. 
	\end{itemize}
	
	At the core of our results is a remarkably simple meta-algorithm for the $(\Delta+1)$ coloring problem: Sample $O(\log{n})$ colors for each vertex independently and uniformly at random from the $\Delta+1$ colors; find a proper coloring of the graph 
	using only the sampled colors. As our main result, we prove that the sampled set of colors with high probability contains a proper coloring of the input graph. The sublinear algorithms are then obtained by designing efficient algorithms 
	for finding a proper coloring of the graph using the sampled colors in each model. 
	
	\smallskip

	All our upper bound results for $(\Delta+1)$ coloring are either optimal or close to best possible in each model studied. We also establish new lower bounds that rule out the possibility of achieving similar results in these 
	models for the closely related problems of maximal independent set and maximal matching. 
	Collectively, our results highlight a sharp contrast between the complexity of $(\Delta+1)$ coloring vs 
	 maximal independent set and maximal matching in various models of sublinear computation even though all three problems are solvable by a simple greedy algorithm in the classical setting. 
\end{abstract}
\setcounter{page}{0}
\clearpage

\thispagestyle{empty}
\setcounter{tocdepth}{3}
\tableofcontents
\clearpage
\setcounter{page}{1}

\section{Introduction}\label{sec:intro}

Graph coloring is a central problem in graph theory and has numerous applications in diverse areas of computer science. A proper $c$-coloring of a graph
$G(V,E)$ assigns a color to every vertex from the palette of colors $\set{1,\ldots,c}$ such that no edge is monochromatic, i.e., has the same color on both endpoints. 
An important and well-studied case of graph coloring problems is the $(\Delta+1)$ coloring problem where $\Delta$ is the maximum degree of the graph. Not only does every graph admit
a $(\Delta+1)$ coloring, remarkably, \emph{any} partial coloring of vertices of a graph can be extended to a proper $(\Delta+1)$ coloring of all vertices: simply pick uncolored vertices in any order and assign a color to a vertex not yet assigned to any of its neighbors;  since the max-degree is $\Delta$, such a color always exists. 

In this paper, we study the $(\Delta+1)$ coloring problem in the context of processing \emph{massive} graphs. 
The aforementioned property of $(\Delta+1)$ coloring problem immediately implies a simple (sequential) greedy algorithm for this problem 
in linear time and space. However, when processing massive graphs, even this algorithm can be computationally prohibitive. This is due to various 
limitations arising in processing massive graphs such as requiring to process the graph in a streaming fashion on a single machine or in parallel across multiple machines due to storage limitations, or 
simply not having enough time for reading the whole input. A natural question is then: 
\begin{quote}
	\emph{Can we design \textbf{\emph{sublinear algorithms}} for $(\Delta+1)$ coloring problem in modern models of computation for processing massive graphs?}
\end{quote}

We answer this fundamental question in the affirmative for several canonical classes of sublinear algorithms including (dynamic) graph streaming algorithms, sublinear time/query algorithms, and massively parallel computation (MPC) algorithms. 
We also prove new lower bounds to contrast the complexity of the $(\Delta+1)$ coloring problem in these models with two other closely related Locally Checkable Labeling (LCL) problems (see~\cite{NaorS93}), namely, the
maximal independent set and the maximal matching\footnote{Another closely related LCL problem is the $(2\Delta-1)$ \emph{edge} coloring problem. 
However, as the output in the edge-coloring problem is linear in the input size, one cannot hope to achieve non-trivial algorithms for this problem in models such as
streaming or sublinear time algorithms, and hence we ignore this problem in this paper.}.

\subsection{Our Contributions}\label{sec:result}

We present new sublinear algorithms for the $(\Delta+1)$ coloring problem which are either the first non-trivial ones  
or significantly improve the state-of-the-art. At the core of these algorithms is a  simple meta-algorithm for this problem that we design in this paper; the sublinear algorithms are then obtained by efficiently
implementing this meta-algorithm in each model separately. 

\smallskip
\emph{A Meta-Algorithm for $(\Delta+1)$ Coloring.} The main approach behind our meta-algorithm is to ``sparsify'' the $(\Delta+1)$ coloring problem to a list-coloring problem with lists/palletes of size $O(\log{n})$ for every vertex. This may sound counterintuitive: while every graph admits a $(\Delta+1)$ coloring that can be found via a simple algorithm, there is no guarantee that it 
also admits a list-coloring with $O(\log{n})$-size lists, let alone one that can be found via a sublinear algorithm. The following key structural result that we prove in this paper, however paves the path for this sparsification. 

\begin{Theorem}[\textbf{Palette-Sparsification Theorem}]\label{thm:intro-meta}
	Let $G(V,E)$ be an $n$-vertex graph with maximum degree $\Delta$. Suppose for any vertex $v \in V$, we sample $O(\log{n})$ colors $L(v)$ from $\set{1,\ldots,\Delta+1}$ independently and uniformly at random. 
	Then with high probability there exists a proper $(\Delta+1)$ coloring of $G$ in which the color for every vertex $v$ is chosen from $L(v)$.
\end{Theorem}

In Result~\ref{thm:intro-meta}, as well as throughout the paper, ``with high probability'' means with probability $1-1/\poly{(n)}$ for some large polynomial in $n$.

Result~\ref{thm:intro-meta} can be seen as a ``sparsification'' result for $(\Delta+1)$ coloring: after sampling $O(\log{n})$ colors for each vertex randomly, the total number of monochromatic edges is only $O(n \cdot \log^{2}{(n)})$ with high probability
(see Section~\ref{sec:sublinear} for a simple proof); at the same time, by computing a proper coloring of $G$ using only these $O(n \cdot \log^{2}{(n)})$ edges---which is promised to exist by Result~\ref{thm:intro-meta}---%
we obtain a $(\Delta+1)$ coloring of $G$. As such, Result~\ref{thm:intro-meta} provides a way of sparsifying the graph into only $\Ot(n)$ edges, while still allowing for recovery of a $(\Delta+1)$ coloring of the original graph. 
This sparsification serves as the central tool in our sublinear algorithms for the $(\Delta+1)$ coloring problem. 

We shall remark that, as stated, Result~\ref{thm:intro-meta} only promise the existence of a coloring (which can be found in exponential time), 
but in fact we show that there is a fast and simple procedure to find the corresponding coloring and that this will also be used by our algorithms in each model.
We also note that the bound of $O(\log{n})$ colors in Result~\ref{thm:intro-meta} is asymptotically optimal (see Proposition~\ref{prop:logn-optimal}).

\smallskip
{\emph{Streaming Algorithms.}} Our Result~\ref{thm:intro-meta} can be used to design a \emph{single-pass} semi-streaming algorithm for the $(\Delta+1)$ coloring problem in the most general setting of graph streams, namely, \emph{dynamic} streams that allow both insertions and deletions of edges (see Section~\ref{sec:streaming} for details). 

\begin{result}\label{thm:intro-streaming}
	There exists a randomized single-pass dynamic streaming algorithm for the $(\Delta+1)$ coloring problem using $\Ot(n)$ space. 	
\end{result}

To our knowledge, the only previous semi-streaming algorithm for $(\Delta+1)$ coloring was the folklore $O(\log{n})$-pass streaming simulation of the classical $O(\log{n})$-round distributed/parallel (PRAM) algorithms for this problem (see, e.g.~\cite{Luby88}). No $o(n^2)$ space single-pass streaming algorithm was known for this problem even in insertion-only streams. 
This state-of-affairs was in fact similar to the case of the closely related maximal matching problem: 
the best known semi-streaming algorithm for this problem on dynamic streams uses $\Theta(\log{n})$ passes~\cite{LattanziMSV11,AhnGM12Linear} 
and it is provably impossible to solve this problem using $o(n^2)$-space in a single pass over a dynamic
stream~\cite{AssadiKLY16} (although this problem is trivial in insertion-only streams). Considering this one 
might have guessed a similar lower bound also holds for the $(\Delta+1)$ coloring problem.  We further prove a lower bound of $\Omega(n^2)$-space on the space complexity of single-pass streaming algorithms for computing a maximal independent set 
even in insertion-only streams (see Theorem~\ref{thm:lower-MIS}). 
Result~\ref{thm:intro-streaming} is in sharp contrast to these results, and shows that $(\Delta+1)$ coloring is provably simpler than both problems in the streaming setting.

\smallskip
\emph{Sublinear Time Algorithms.} There exists a straightforward greedy algorithm that computes a $(\Delta+1)$ coloring of any given graph in linear time, i.e., $O(m+n)$ time. 
Perhaps surprisingly, we show that one can improve upon the running time of this textbook algorithm by using Result~\ref{thm:intro-meta}. 

\begin{result}\label{thm:intro-subtime}
	There exists a randomized $\Ot(n\sqrt{n})$ time algorithm for the $(\Delta+1)$ coloring problem. Furthermore, any algorithm for this problem requires $\Omega(n\sqrt{n})$ time. 
\end{result}

When designing sublinear (in $m$) time algorithms, specifying the exact data model is important as the algorithm cannot even read the entire input once. 
In Result~\ref{thm:intro-subtime}, we assume the standard query model for sublinear time algorithms on general graphs (see, e.g., Chapter 10 of Goldreich's 
book~\cite{Goldreich17}) which allow for two types of queries $(i)$ what is the $i$-th neighbor of a given vertex $v$, and $(ii)$ whether a given pair of vertices $(u,v)$ are neighbor to each other or not (see Section~\ref{sec:sub-time} for details). 
To our knowledge, this is the first sublinear time algorithm for the $(\Delta+1)$ coloring problem. We also note that an important feature of our algorithm in Result~\ref{thm:intro-subtime} is that it is \emph{non-adaptive}, i.e., it chooses
 all the queries to the graph beforehand and thus queries are done in parallel. 

In yet another contrast to the $(\Delta+1)$ coloring problem, we show that the problem of computing a maximal matching requires $\Omega(n^2)$ queries to the graph 
and hence $\Omega(n^2)$ time  (see Theorem~\ref{thm:lower-matching}). 

\smallskip
\emph{Massively Parallel Computation Algorithms.} Another application of our Result~\ref{thm:intro-meta} is a \emph{constant-round} algorithm for the $(\Delta+1)$ coloring problem in the MPC model, which is a common
abstraction of MapReduce-style computation frameworks (see Section~\ref{sec:mpc} for formal definition). 

\begin{result}\label{thm:intro-mpc}
	There exists a randomized MPC algorithm for the $(\Delta+1)$ coloring problem in at most two MPC rounds on machines with memory $\Ot(n)$. 
\end{result}

Two recent papers considered graph coloring problems in the MPC model. 
Harvey~\etal~\cite{HarveyLL18} designed algorithms that use $n^{1+\Omega(1)}$ memory per machine and find a $(\Delta+o(\Delta))$ coloring of a given graph---an algorithmically easier problem than $(\Delta+1)$ coloring---in $O(1)$ MPC 
rounds. Furthermore, Parter~\cite{Parter18} designed an MPC algorithm that uses $O(n)$ memory per machine and finds a $(\Delta+1)$ coloring in $O(\log\log{\Delta}\cdot\log^{*}\!{(n)})$ rounds\footnote{The algorithm of Parter~\cite{Parter18} is stated 
in the Congested-Clique model, but using the well-known connections between this model and the MPC model, see, e.g.~\cite{BehnezhadDH18,GhaffariGMR18}, this algorithm immediately extends to the MPC model.}. 
Our Result~\ref{thm:intro-mpc} improves these results significantly: both the number of used colors as well as per machine memory compared to~\cite{HarveyLL18}, and round-complexity (with at most $\polylog{(n)}$-factor more per-machine memory) compared to~\cite{Parter18}.

Maximal matching and maximal independent set problems have also been studied previously in the MPC model~\cite{LattanziMSV11,GhaffariGMR18,Konrad18}. Currently, the best known algorithms for these problem
with $\Ot(n)$ memory per machine require $O(\log{n})$ rounds in case of maximal matching~\cite{LattanziMSV11} and $O(\log\log{n})$ rounds in case of maximal independent set~\cite{GhaffariGMR18,Konrad18}. 
For the related problems of $O(1)$-approximating the maximum matching and the minimum vertex cover, a recent set of results achieve $O(\log\log{n})$-round MPC algorithms
with $\Ot(n)$ memory per machine~\cite{Assadi17,AssadiBBMS18,GhaffariGMR18,CzumajLMMOS18} (these results however do not extend to the maximal matching problem). 
Our Result~\ref{thm:intro-mpc} hence is the first example that gives a constant round MPC algorithm for one of 
the ``classic four local distributed graph problems'' (i.e., maximal independent set, maximal matching, $(\Delta+1)$ vertex coloring, and $(2\Delta-1)$ edge coloring; see, e.g.~\cite{PanconesiR01,BarenboimE11,FischerGK17}), 
even when the memory per machine is as small as $\Ot(n)$. 

\paragraph{Optimality of Our Sublinear Algorithms:} Space-complexity of our streaming algorithm in Result~\ref{thm:intro-streaming} and round-complexity of our MPC algorithm in Result~\ref{thm:intro-mpc} are 
clearly optimal (to within poly-log factors and constant factors, respectively). We further prove that query and time complexity 
of our sublinear time algorithm in Result~\ref{thm:intro-subtime} are also optimal up to poly-log factors (see Theorem~\ref{thm:lower-color}). 

\paragraph{Perspective: Beyond Greedy Algorithms.}
Many graph problems 
admit simple greedy algorithms. Starting with Luby's celebrated distributed/parallel algorithm for the maximal independent set problem~\cite{Luby85},
there have been numerous attempts in adapting these greedy algorithms to different models of computation including the models considered in this paper (see, e.g.~\cite{LattanziMSV11,KumarMVV13,IndykMMM14,HarveyLL18,NguyenO08,AlonRVX12}). 
Typically these adaptations require multiple passes/rounds of computation, and this is for the fundamental reason that most greedy algorithms are inherently sequential: they require 
accessing the input graph in an \emph{adaptive} manner based on decisions made thus far, which, although limited, still results in requiring \emph{multiple passes/rounds} over the input.

Our work on $(\Delta+1)$ coloring bypasses this limitation of greedy algorithms by utilizing a completely different approach, namely, a non-adaptive sparsification of the input graph (Result~\ref{thm:intro-meta}) that in turns lends itself to space, time, and 
communication-efficient algorithms in a variety of different models. As such, our results can be seen as an evidence that directly adapting greedy algorithms for graph problems may not necessarily be the best choice in these models.
We believe that this viewpoint is an important (non-technical) contribution of our paper as it may pave the way for obtaining more efficient algorithms for other fundamental graph 
problems in these models. 

\subsection{Our Techniques}\label{sec:techniques}

 The main technical ingredient of our paper is Result~\ref{thm:intro-meta}. For intuition, consider two extreme cases: when
the underlying graph is very dense, say is a clique on $\Delta+1$ vertices, 
and when the underlying graph is relatively sparse, say every vertex (except for one) have degree at most $\Delta/2$. Result~\ref{thm:intro-meta} is easy to prove for either case albeit by using 
entirely different arguments. For the former case, consider the bipartite graph consisting of vertices in $V$ on one side and set of colors $\set{1,\ldots,\Delta+1}$ on the other side, where each vertex $v$ in the $V$-side is connected to vertices in 
$L(v)$ in the color-side. Using standard results from random graphs theory, one can argue that this graph with high probability has a perfect matching, thus implying the list-coloring of $G$. For the latter case, 
consider the following simple greedy algorithm: iteratively sample a color for every vertex from the set $\set{1,\ldots,\Delta+1}$ and assign the color to the vertex if it is not chosen by any of its neighbors so far. 
It is well-known that this algorithm only requires $O(\log{n})$ rounds when number of colors is a constant factor larger than the degree (see, e.g.,~\cite{SchneiderW10}). As such, the set of colors sampled in the list $L(v)$ for vertices $v \in V$ is enough to 
``simulate'' this algorithm in this case. 

To prove Result~\ref{thm:intro-meta} in general, we need to interpolate between these two extreme cases. 
To do so, we exploit a graph decomposition result of~\cite{HarrisSS16} (see also~\cite{ChangLP18}) for the $(\Delta+1)$ coloring problem, that allows
for decomposing a graph into ``sparse'' and ``dense'' components. The proof for coloring the sparse components then more or less follows by simulating standard distributed algorithms in~\cite{SchneiderW10,ElkinPS15} as discussed above. 
The main part however is to prove the result for dense components which requires a global and non-greedy argument. Note that in general, we can always 
reduce the problem of finding a $(\Delta+1)$ coloring to an instance of the assignment problem on the bipartite graph $V \times \set{1,\ldots,\Delta+1}$ discussed above. The difference is that we need to allow some vertices
in $\set{1,\ldots,\Delta+1}$ to be assigned to more than one vertex in $V$ when $\card{V} > \Delta+1$ (as opposed to the case of cliques above that only required finding a perfect matching). 
We show that if the original graph is ``sufficiently close'' to being a clique, then with high probability, such an assignment exists in this bipartite graph and use this to prove the existence of the desired list-coloring of $G$. 

Result~\ref{thm:intro-meta} implies the sublinear algorithms we design in each model with a simple caveat: The list-coloring problem that needs to be solved in the sparsified graph is in general NP-hard and hence using this result directly
does not allow for a polynomial time implementation of our algorithms. We thus combine Result~\ref{thm:intro-meta} with additional ideas specific to each model to turn these algorithms into polynomial time (and in fact even sublinear time) algorithms. 

\subsection{Recent Related Work}\label{sec:recent}

Independently and concurrently to our work, two other papers also considered the vertex coloring problem in settings related to this paper. 
Firstly, Parter and Su~\cite{ParterS18}, improving upon the previous algorithm of Parter~\cite{Parter18}, 
gave an $O(\log^{*}\!{(\Delta)})$ round congested-clique algorithm for $(\Delta+1)$ coloring; this result also immediately implies an MPC algorithm for $(\Delta+1)$ coloring in $O(\log^{*}\!{(\Delta)})$ rounds and $O(n)$ memory per-machine. 
Moreover, Bera and Ghosh~\cite{BeraG18} also studied the graph coloring problem in the streaming model and gave a single-pass streaming algorithm that for any parameter $\eps \in (0,1)$, outputs a $(1+\eps)\Delta$ coloring of the input graph
using $\Ot(n/\eps)$ space. Note that for the $(\Delta+1)$ coloring problem, this algorithm requires $\Omega(n\Delta)$ space which is equal to the input size. 

Subsequent to our work, Chang~\etal~\cite{ChangFGUZ18} further 
studied the $(\Delta+1)$ coloring problem and among other results, gave an $O(\sqrt{\log\log{n}})$ round MPC algorithm for this problem on machines with memory as small as $n^{\Omega(1)}$.

\newcommand{\Vdense}{\ensuremath{V^{\textnormal{dense}}}\xspace}

\newcommand{\Vsparse}{\ensuremath{V^{\textnormal{sparse}}}\xspace}
\section{Preliminaries}\label{sec:prelim}

\paragraph{Notation.} For any $t \geq 1$, we define $[t] := \set{1,\ldots,t}$. For a graph $G(V,E)$, we use $V(G) := V$ to denote the vertices, $E(G) :=E$ to denote the edges, and $\deg{(v)}$ to denote the 
degree of $v \in V$. 

\subsection{The Harris-Schneider-Su (HSS) Network Decomposition}\label{sec:HSS-decomposition}

In the proof of our Result~\ref{thm:intro-meta}, we use a network decomposition result of Harris, Schneider and Su for designing distributed algorithms for graph coloring in the LOCAL model~\cite{HarrisSS16}. 
We emphasize that we use of this decomposition in a quite different way than the ones in distributed settings~\cite{HarrisSS16,ChangLP18}.

The Harris-Schneider-Su network decomposition, henceforth HSS-decomposition, partitions a graph $G(V,E)$ into \emph{sparse} and \emph{dense} regimes, measured with respect to a parameter $\eps \in [0,1)$. 

\begin{definition}[$\eps$-friend edges]\label{def:eps-friend}
	For any $\eps \in [0,1)$, we say that an edge $(u,v)$ in a graph $G$ is an \emph{$\eps$-friend} edge iff $\card{N(u) \cap N(v)} \geq (1-\eps) \Delta$. Let $F_{\eps} \subseteq E(G)$ denote the set of 
	all $\eps$-friend edges. 
\end{definition}

\begin{definition}[$\eps$-dense vertices]\label{def:eps-dense}
	For any $\eps \in [0,1)$ we say that a vertex $v$ in a graph $G$ is \emph{$\eps$-dense} iff degree of $v$ in $F_{\eps}$ is at least $(1-\eps)\Delta$.
	We use $\Vdense_{\eps}$ to denote the set of all $\eps$-dense vertices. 
\end{definition}

Consider the graph $H_{\eps} := H(\Vdense_{\eps},F_{\eps})$ as the subgraph of $G$ on the set of $\eps$-dense vertices and containing only the $\eps$-friend edges. Let 
$C_1,\ldots,C_k$ be the connected components of $H_{\eps}$. HSS-decomposition partitions the vertices of the graph into $\Vdense_\eps$ and $V \setminus \Vdense_\eps$, where $\Vdense_\eps$ is partitioned into $C_1,\ldots,C_k$ with the following properties given
by Lemma~\ref{lem:HSS-decomposition} and Proposition~\ref{prop:sparse-vertices}.

\begin{lemma}[\!\!\cite{HarrisSS16}]\label{lem:HSS-decomposition}
	Any connected component $C_i$ of $H_{\eps}$ has size at most $(1+3\eps)\Delta$. Moreover, any vertex $v \in C_i$ has at most:
	\begin{enumerate}
	\item $\eps\Delta$ neighbors (in $G$) in $\Vdense_{\eps} \setminus C_i$, i.e., $\card{N_G(v) \cap (\Vdense_{\eps} \setminus C_i)} \leq \eps \Delta$. 
	\item $3\eps\Delta$ non-neighbors (in $G$) in $C_i$, i.e., $\card{C_i \setminus N_G(v)} \leq 3\eps\Delta$.
	\end{enumerate}
\end{lemma}

Define \emph{$\eps$-sparse} vertices as $\Vsparse_\eps := V \setminus \Vdense_{\eps}$, i.e., the vertices which are not $\eps$-dense. The main
property of $\eps$-sparse vertices we are interested in is as follows.
\begin{proposition}\label{prop:sparse-vertices}
	Let $v$ be any $\eps$-sparse vertex in $G$. Then, the total number of edges spanning the neighborhood of $v$ is at most $(1-\eps^2)\cdot{{\Delta}\choose{2}}$, that
	is $\card{\set{(u,w) \mid u \in N_G(v) \wedge w \in N_G(v)}} \leq (1-\eps^2)\cdot{{\Delta}\choose{2}}$. 
\end{proposition}
\begin{proof}
	If $\deg(v)$ is less than $\Delta$, to prove the proposition, we add some dummy vertices which are only connected to $v$ so that $\deg(v)$ become exactly $\Delta$, By doing so, the number of edges spanning the neighborhood of $v$ would not change.
	As $v$ is an $\eps$-sparse vertex, it means that at least $\eps \cdot \Delta$ of its neighbors have at most $(1-\eps) \cdot \Delta$ neighbors in common with $v$. This means that any of those vertices is \emph{not} connected to 
	at least $\eps \cdot \Delta$ other vertices in $N(v)$. As such, the total number of edges spanning the neighborhood of $v$ is at most, 
	${{\Delta}\choose{2}} - \frac{1}{2} \cdot \eps^2 \cdot \Delta^2 \leq (1-\eps^2) \cdot {{\Delta}\choose{2}}$.
\end{proof}

To conclude, HSS-decomposition partitions the vertices of the graph into $\Vsparse_\eps \cup \Vdense_\eps$, where $\Vdense_\eps$ is additionally partitioned into the collection $C_1,\ldots,C_k$ with the properties given
by Lemma~\ref{lem:HSS-decomposition} and Proposition~\ref{prop:sparse-vertices}. 

\subsection{A Simple Extension of the HSS-Decomposition}

It would be more convenient for us to work with a slightly different variant of the HSS-decomposition that we introduce here. 

\begin{lemma}[Extended HSS-Decomposition]\label{lem:extended-HSS-decomposition}
	For any parameter $\eps \in [0,1)$, any graph $G(V,E)$ can be decomposed into a collection of vertices $V:= \Vsparse_\star \cup C_1 \cup \ldots \cup C_k$ such that: 
	\begin{enumerate}[leftmargin=15pt]
		\item\label{HSS-p1} $\Vsparse_{2\eps} \subseteq \Vsparse_\star \subseteq \Vsparse_{\eps}$, i.e., any vertex in $\Vsparse_\star$ is at least $(2\eps)$-sparse and at most $\eps$-sparse. 
		\item\label{HSS-p2} For any $i \in [k]$, $C_i$ has the following properties (we refer to $C_i$ as an almost-clique): 
		\begin{enumerate}
			\item\label{HSS-p2a}\label{p1} $(1-\eps)\Delta \leq \card{C_i} \leq (1+6\eps)\Delta$. 
			\item\label{HSS-p2b} Any $v \in C_i$ has at most $7\eps\Delta$ neighbors outside of $C_i$.
			\item\label{HSS-p2c} Any $v \in C_i$ has at most $6\eps\Delta$ non-neighbors inside of $C_i$. 
		\end{enumerate}
	\end{enumerate}
\end{lemma}
\noindent
Two main differences between Lemma~\ref{lem:HSS-decomposition} and the original HSS-decomposition are: $(i)$ size of each $C_i$ is now lower bounded (HSS-decomposition does not lower bound the size of $C_i$), and $(ii)$ the number of all neighbors of any vertex outside $C_i$ is now bounded (not only neighbors to other dense vertices as in the original HSS-decomposition). 

\begin{proof}[Proof of Lemma~\ref{lem:extended-HSS-decomposition}]
    Consider the HSS-Decomposition with parameter $2\eps$. By Lemma \ref{lem:HSS-decomposition}, $G$ can be decomposed into $2\eps$-sparse vertices $\Vsparse_{2\eps}$ and 
    components $C_1,\ldots,C_\ell$ with $2\eps$-dense vertices. Let $C_1,C_2,\dots,C_k$ be the components among these that contain at least one $\eps$-dense vertex. 
    
    We define $\Vsparse_\star$ as the set of vertices in $\Vsparse_{2\eps} \cup (\Vdense_{2\eps} \setminus C_1 \cup \ldots \cup C_k)$, i.e., all vertices that are not in the $k$ connected components defined above. Clearly,
    $\Vsparse_{2\eps} \subseteq \Vsparse_\star$. On the other hand, $\Vsparse_\star$ does not contain any $\eps$-dense vertices (as we removed $C_1,\ldots,C_k$), and hence $\Vsparse_\star \subseteq \Vsparse_{2\eps}$.
    This proves Property~(\ref{HSS-p1}). We now prove Property~(\ref{HSS-p2}).
    
    Fix any $i \in [k]$ and let $C_i$ be any connected component that contains a $\eps$-dense vertex. Firstly, since $C_i$ is a connected component of a HSS-decomposition with parameter $2\eps$, 
    by Lemma~\ref{lem:HSS-decomposition}, any vertex in $C_i$ has at most $6\eps\Delta$ non-neighbors inside $C_i$. This proves Property~(\ref{HSS-p2c}). 
    
    Now let $v$ be any $\eps$-dense vertex in $C_i$. As $v$ is $\eps$-dense, by Definition~\ref{def:eps-dense}, $v$ has at least $(1-\eps)\Delta$ $\eps$-friend neighbors.
    Let $S_v$ be the set of these vertices. By Definition~\ref{def:eps-friend}, any of these vertices have at least $(1-\eps)\Delta$ shared neighbors with $v$. As the maximum degree of any vertex is $\Delta$, this implies
    that any two vertices $u,w \in S_v$ have at least $(1-2\eps)\Delta$ common neighbors with each other. Furthermore, since $S_v$ has at least $(1-\eps)\Delta$ vertices, each vertex in $S_v$ has at least $(1-2\eps)\Delta$ neighbors in $S_v$. Thus  all vertices in $S_v$ are $(2\eps)$-dense. 
    Moreover, as all vertices in $S_v$ are connected to $v$ by an $\eps$-friend edge (and hence also a $2\eps$-friend edge), vertices in $S_v$ all appear in the same connected component $C_i$ with
    the vertex $v$.  This implies that $\card{C_i} \geq \card{S_v} \geq (1-\eps)\Delta$. Moreover, by Property~(\ref{HSS-p2c}) we already proved, any vertex has at most $6\eps\Delta$ non-neighbors in $C_i$ and hence
    $\card{C_i} \leq (1+6\eps)\Delta$. This proved Property~(\ref{HSS-p2a}).
    
    Finally, the above argument, together with the lower bound on size of $C_i$, also implies that each vertex $u \in C_i$ is connected to at least $(1-7\eps)\Delta$ vertices inside $C_i$. As such, $u$ can only
    have $7\eps\Delta$ neighbors outside $C_i$ proving Property~(\ref{HSS-p2b}).   
\end{proof}

\newcommand{\bd}{\ensuremath{\bar{d}}}
\renewcommand{\bC}{\ensuremath{\overline{C}}}
\newcommand{\hC}{\ensuremath{\widehat{C}}}

\newcommand{\Vr}{\ensuremath{V^{\textnormal{\textsf{rem}}}}}
\newcommand{\Gr}{\ensuremath{G^{\textnormal{\textsf{rem}}}}}

\newcommand{\drem}[1]{\ensuremath{d^{\textnormal{\textsf{rem}}}(#1)}}

\newcommand{\GC}{\ensuremath{\textnormal{\textsf{GreedyColor}}}\xspace}

\newcommand{\col}{\ensuremath{\mathcal{C}}}

\section{The Palette-Sparsification Theorem}\label{sec:color-sampling}

We prove our Result~\ref{thm:intro-meta} in this section; see Appendix~\ref{app:sparsification} for further remarks on optimality of the bounds in this result, as well as (im)possibility of extending this result to $c$-coloring for 
values of $c$ strictly smaller than $\Delta+1$. 

\begin{theorem}[\textbf{Palette-Sparsification Theorem}]\label{thm:color-sampling}
    Let $G(V,E)$ be any $n$-vertex graph and $\Delta$ be the maximum degree in $G$. Suppose for each vertex $v \in V$,
    we independently pick a set $L(v)$ of colors of size $\Theta(\log{n})$ uniformly at random from $[\Delta+1]$. Then 
    with high probability there exists a proper coloring $\col: V \rightarrow [\Delta+1]$ of $G$ such that for all vertices $v \in V$, $\col(v) \in L(v)$. 
\end{theorem}

Let us start by fixing the parameters used in the proof of Theorem~\ref{thm:color-sampling}. Let $\eps > 0$ be a sufficiently small constant, say, $\eps := 1/5000$ and $\alpha > 0$ be a sufficiently large integer, say $\alpha = 5000$\footnote{In the 
interest of simplifying the exposition of the proof, we made no attempt in optimizing the constants. The proof of the theorem can be made to work with much smaller constants than the ones used here.}. 
In Theorem~\ref{thm:color-sampling}, we make each vertex sample $\paren{\alpha\cdot\log{n}/\eps^2}$ colors in $L(v)$. 
We assume that $(\Delta+1) > \alpha\cdot\log{n}/\eps^2$; otherwise Theorem~\ref{thm:color-sampling} is trivial as we sampled all $\Delta+1$ colors for each vertex and every graph admits a $(\Delta+1)$ coloring. 
For the purpose of the analysis, we assume that the set $L(v)$ of each vertex is union of three sets $L_1(v) \cup L_2(v) \cup L_3(v)$, named \emph{batch} one, two, and three, respectively, where
each $L_i(v)$ for $i \in [3]$ is created by picking each color in $[\Delta+1]$ independently and with probability $p := \Paren{\frac{\alpha \cdot \log{n}}{3\eps^2 \cdot (\Delta+1)}}$. 
While this distribution is not identical to the one in Theorem~\ref{thm:color-sampling}, it is easy to see that proving the theorem for this distribution also implies the original result as in this new distribution, with high probability, no vertex
samples more than $O(\log{n})$ colors. 

We use the extended HSS-decomposition with parameter $\eps$ (Lemma~\ref{lem:extended-HSS-decomposition}): 
graph $G$ is decomposed into $V := \Vsparse_\star \cup C_1 \cup \ldots \cup C_k$ where each $C_i$ for $i \in [k]$ is an almost-clique. 

We prove Theorem~\ref{thm:color-sampling} in three parts. 
In the first part, we argue that by only using the colors in the first batch $L_1(\cdot)$, we can color all the vertices in $\Vsparse_{\star}$. This part is mostly standard and more or less follows
from the results in~\cite{ElkinPS15,HarrisSS16,ChangLP18} by simulating a distributed local algorithm using only the colors in the first batch. 
We hence concentrate bulk of our effort in proving the next two parts which are the key components of the proof. We first show that using 
only the colors in the second batch, we can color a relatively large fraction of vertices in each almost-clique $C_i$ \emph{at a rate of two vertices per color} (assuming the number of non-edges in the almost-clique is not too small). This allows 
us to ``save'' extra colors for coloring the remainder of the almost-cliques, which we do in the last part. We note that unlike the coloring of the first part which is based on a ``local'' coloring scheme (in which we determine the color of each vertex based
on colors assigned to each of its neighbors similar to the greedy algorithm), the coloring of the second and third part is done in a ``global'' manner in which the color of a vertex is determined based on some global properties of the graph not only
the local neighborhood of a vertex. 

\smallskip

\emph{Partial Coloring Function.} Define a function $\col: V \rightarrow [\Delta+1] \cup \set{\perp}$ that assigns one of the colors in $[\Delta+1]$ plus the \emph{null color} $\perp$ to the vertices, such that no two neighboring vertices have the same color 
from $[\Delta+1]$ (but they may both have the null color $\perp$). We refer to $\col$ as a \emph{partial coloring} function and refer to vertices that are colored by $\col$ in $[\Delta+1]$ as having a \emph{valid color}. Furthermore, we say that a valid color $c$ is \emph{available} to a vertex $v \in V$ in the partial coloring $\col$, iff $\col$ does not assign $c$ to any neighbor of $v$. 
The set of available colors for $v$ is denoted by $A_\col(v)$. 

It is immediate that if $\col$ does not assign a null color to any vertex, then the resulting coloring is a proper $(\Delta+1)$-coloring of the graph. We start with a partial coloring function $\col$ which assigns a null color to all vertices initially and modify this coloring in each part to remove all null colors. 

\subsection{Part One: Coloring Sparse Vertices.} 
Recall the definition of sparse vertices $\Vsparse_\star$ in the extended HSS-decomposition from Section~\ref{sec:prelim}. In the first part of the proof, we show that we can color all sparse vertices using only the colors in the first batch. 

\begin{lemma}\label{lem:sparse-color}
    With high probability, there exists a partial coloring function $\col_1:V \rightarrow [\Delta+1] \cup \set{\perp}$ such that for all vertices $v \in \Vsparse_\star$, $\col_1(v) \in L_1(v)$.  
\end{lemma}

We prove this lemma by showing that one can simulate a natural greedy algorithm (similar but not identical to the algorithm of~\cite{ElkinPS15}) for coloring sparse vertices using only the colors in the first batch.  
The first step is to use the \emph{first} color in $L_1(v)$, chosen uniformly at random from $[\Delta+1]$, 
for all vertices $v \in \Vsparse_\star$ to color a large fraction of vertices in $\Vsparse_\star$; the main property of this coloring is that after this step any uncolored 
$\eps$-sparse vertex has $\Omega(\eps^2\Delta)$ ``excess'' colors compared to the number of edges it has to other \emph{uncolored} $\eps$-sparse vertices. 
This step follows from the proof of the \textsf{OneShotColoring} procedure in~\cite{ElkinPS15,HarrisSS16,ChangLP18} and we simply present a proof sketch for intuition.  
We then use the remaining colors in $L_1(\cdot)$ for each uncolored vertex and color them greedily, using the fact that the number of available colors is sufficiently larger than the number of 
neighbors of each uncolored vertex in every step. This part is also similar to the algorithm in~\cite{ElkinPS15} (see also~\cite{HarrisSS16,ChangLP18}) 
but uses a different argument as here we \emph{cannot} sample the colors for each vertex \emph{adaptively} (as the colors in
$L_1(\cdot)$ are  sampled ``at the beginning'' of the greedy algorithm).

As the proof of this lemma closely follows the previous work in~\cite{ElkinPS15,HarrisSS16,ChangLP18} with only some minor modifications, we postpone its proof to Appendix~\ref{app:sparse}.

\subsection{Part Two: Initial Coloring of Almost Cliques.}
Recall that by Lemma~\ref{lem:sparse-color}, after the first part, we have a partial coloring $\col_1: V \rightarrow [\Delta+1] \cup \set{\perp}$ that assigns a valid color to all sparse vertices. 
We now design a partial coloring $\col_2 : V \rightarrow [\Delta+1] \cup \set{\perp}$ 
where for all $v \in \Vsparse_\star$, $\col_2(v) := \col_1(v)$ and $\col_2(v) = \perp$ for remaining vertices initially but some additional vertices  would also be assigned a valid color by the end of this part using the second batch. 

Fix the almost-cliques $C_1,\ldots,C_k$. Define $\bC_i$ as the \emph{complement-graph} of $C_i$ on the same set of vertices as $C_i$ by switching edges and non-edges in $C_i$. 
Note that any two neighboring vertices in $\bC_i$ \emph{can be colored the same} (in $G$). We exploit this fact 
in the following definition. 

\begin{definition}[Colorful Matching]\label{def:colorful}
    We say that a matching $M_i$ in the complement-graph $\bC_i$ of an almost-clique is a \emph{colorful matching} with respect to the partial coloring $\col_2$ iff: 
    \begin{enumerate}
        \item For any $(u,v) \in M_i$ there is a color $c_{u,v}$ s.t $c_{u,v} \in L_2(u) \cap L_2(v)$ and $c_{u,v} \in A_{\col_2}(u) \cap A_{\col_2}(v)$.
        \item For any pairs of edges $(u_1,v_1),(u_2,v_2) \in M_i$, $c_{u_1,v_1} \neq c_{u_2,v_2}$.
    \end{enumerate}
\end{definition}

By finding a colorful matching in a complement-graph $\bC_i$, we can ``save'' on the the colors needed for coloring $C_i$ as we can color vertices of the matching at a rate of two vertices per color. 

We now iterate over complement-graphs $\bC_1,\ldots,\bC_k$ one by one, and show that there exists a sufficiently large colorful matching in each complement-graph, \emph{even after} we update the coloring $\col_2$ for vertices matched by 
the colorful matchings in previous 
complement-graphs. 

\begin{lemma}\label{lem:colorful}
    Fix any complement-graph $\bC$ and let $\col_2 : V \rightarrow [\Delta+1] \cup \set{\perp}$ be any partial coloring in which $c(v) = \perp$ for all $v \in \bC$. Suppose \emph{average degree} of $\bC$ is $\bd$. Then, there exists a colorful
    matching of size at least $\paren{\bd/(320\eps)}$ in $\bC$ with high probability (over the randomness of $L_2$).
\end{lemma}

We start by some definitions. For $(u,v)\in \bar{C}$, a color is {\em available} to this edge if the color is available to both $u$ and $v$ under $\col_2$. 
For a set of colors $D$, let $a_D(e)$ be the number of available colors for an edge $e$ in $D$. For a set of edges $F$, we define $a_D(F) := \sum_{e\in F}a_D(e)$. 
We say that an edge $e=(u,v)$ sampled an available color in $L_2$ iff there exists an available color $c$ for $e$ in $L_2(u) \cap L_2(v)$. 
Lemma~\ref{lem:colorful} relies on the following lemma. 

\begin{lemma}\label{lem:colorvar}
    Let $\bar{C'}$ be a subgraph of $\bar{C}$ and $F = E(\bar{C'})$ be its edge-set. Let $D$ be any set of colors such that $a_D(F)\ge 120\epsilon^2\Delta^2$. If for each vertex in $\bar{C'}$, we sample each color in $D$ with probability $p=\frac{20\log{n}}{\epsilon\Delta}$, then with high probability, there is an edge $(u,v)$ in $F$ that samples an available color. 
\end{lemma}
\begin{proof}
    Let $q := \frac{1}{\eps\Delta}$. We argue that if each vertex samples each color in $D$ with probability $q$, then with a constant probability, there is
     an edge $(u,v)$ in $F$ that samples an available color. We then argue that sampling with rate $p$ can be seen
    as performing this experiment independently $O(\log{n})$ times and obtain the final high probability bound.
    
    For an edge $e=(u,v)$, let $X_{e}$ be an indicator random variable which is $1$ iff $e$ sampled an available color (in the experiment with probability $q$). Since $q^2 \cdot a_D(e) \le \frac{\Delta}{\epsilon^2\Delta^2} < 1$, we have, 
    \begin{align*}
    	\Ex\bracket{X_e} &= 1-\Pr\bracket{X_e = 0} = 1- \paren{1-q^2}^{a_D(e)} \\
	&\geq 1 - \exp\paren{-q^2 \cdot a_D(e)} \geq \frac{q^2 \cdot a_D(e)}{2}. \tag{as $1-x \leq e^{-x} \leq 1-x/2$ for $x < 1$}
    \end{align*} 
    
   Define $X = \sum_{e\in F} X_e$. We prove  $\prob{X>0} \geq \frac{1}{15}$ which implies that with probability at least $1/15$, an edge in $F$ samples an available color.

    Firstly, notice that $\expect{X}=\sum_{e\in F} \expect{X_e}\ge \frac{a_D(F) \cdot q^2}{2} \geq 120\eps^2\cdot\Delta^2 \cdot \frac{1}{2\eps^2\Delta^2} = 60$. 
    We prove that the variance of $X$ is not much larger than its expectation, and use Chebyshev's inequality to prove the bound on $\Pr\paren{X > 0}$. 
    By definition, $\var{X} = \sum_{e \in F} \var{X_e} + \sum_{e \neq e' \in F} \cov{X_e,X_{e'}}$. Since each $X_e \in \set{0,1}$, we have
    $\var{X_e} \leq \expect{X_e}$, hence it only remains to bound the covariance terms.

    For any pair of edges $e, e'$ in $F$, if they do not share a common endpoint, then the variables $X_e$ and $X_{e'}$ are independent (hence their covariance is $0$), 
    but if they share a common endpoint, their covariance would be non-zero. However, in this case, $\cov{X_e,X_{e'}} \leq \expect{X_e\cdot X_{e'}}\le 1-\paren{1-q^3}^{a_D(e)} \le 1 - \exp\paren{-2q^3 \cdot a_D(e)} \le 2 a_D(e) \cdot q^3$. 
    By Property~(\ref{HSS-p2c}) of Lemma~\ref{lem:extended-HSS-decomposition}, each vertex in $\bar{C}$ has at most $6\eps\Delta$ neighbors (as edges in $\bar{C}$ are non-edges in the almost-clique $C$).
    As such, there are at most $12\epsilon\Delta$ edges that share a common endpoint with an edge $e$. Let $S(e)$ denote the set of edges in $F$ that share an endpoint with $e$. We have, 
    \begin{align*}
        \var{X} &\le \sum_{e\in F} \expect{X_e} + \sum_{e\in F} \sum_{e' \in S(e)} 2a_D(e)q^3 \\
        &\le \expect{X} + \sum_{e\in F} a_D(e) \cdot 24\epsilon\Delta q^3 \le 50\expect{X}. 
    \end{align*}
    The last equation is because $q=\frac{1}{\epsilon\Delta}$. By Chebyshev's inequality: $\Pr\bracket{X=0} \leq \frac{\var{X}}{\expect{X}^2} \le \frac{5}{6}$.

    So if we sample each color with probability $q$, there is an edge $(u,v)$ that samples an available color with probability at least $\frac{1}{6}$. By sampling the colors 
    at rate $p = 20\log{n} \cdot q$, we can repeat this trial at least $10\log n$ times and obtain that with $1-(1/6)^{10\log{n}} \geq 1-1/n^{3}$, 
    there is an edge that sampled an available color. 
\end{proof}

\noindent
We are now ready to prove Lemma~\ref{lem:colorful}. 

\begin{proof}[Proof of Lemma \ref{lem:colorful}]
    We construct the colorful matching in the following manner. 
   We iterate over the colors (in an arbitrary order) and for each color $c$, we find the vertices which sampled this color in $L_2(\cdot)$ (this choice is independent across colors by the sampling process that defines $L_2(\cdot)$). 
    If $c$ is available for some edge $(u,v)$ in $\bar{C}$, we add $(u,v)$ with color $c_{u,v}=c$ to the matching, delete this edge from the graph, and move on to the next color. Clearly the resulting matching will be colorful (as in Definition~\ref{def:colorful}). 
    It thus remains to lower bound the size of this matching. 
    
    Let $\bar{C'}$ be $\bar{C}$ initially and $F$ be its edge-set, i.e., $F = E(\bar{C'})$. $D$ is also initially the set of all colors in $[\Delta+1]$. 
    Let $N$ be the number of vertices in $\bar{C}$. 
    Consider the value of $a_D(F)$ throughout the process. When we are dealing with a color $c$, if we cannot find an edge $(u,v) \in \bar{C'}$ where $c$ is available for
    $(u,v)$, we delete the color $c$ from $D$. In this case, $a_D(F)$ will decrease by $a_c(F) \leq \card{F}$. 
    Otherwise, we add $(u,v)$ with color $c$ to our colorful matching, delete 
    $c$ from $D$, and delete $u$ and $v$ from $\bar{C'}$. In this case, $a_D(F)$ will decrease by at
    most $a_c(F)+12\eps\Delta^2 \le \bar{d} \Delta + 12\eps \Delta^2 \le 18 \eps \Delta^2$ since each vertex in $\bar{C'}$ has at most $6\eps \Delta$ neighbors (by Property~(\ref{HSS-p2c}) of extended HSS-decomposition in Lemma~\ref{lem:extended-HSS-decomposition}) and $a_c(F)$ 
    is at most $\card{F} = \bar{d}N/2 \leq \bar{d} \cdot (1+6\eps)\Delta/2 \leq \bar{d}\Delta$ as in the extended
    HSS-decomposition, $N = \card{\bar{C}} \leq (1+6\eps)\Delta$ (by Property~(\ref{HSS-p2a}) of Lemma~\ref{lem:extended-HSS-decomposition}). By Lemma \ref{lem:colorvar} (as the process of sampling colors in $L_2(\cdot)$ is identical to the lemma statement but sampling colors with higher probability), with high probability, $a_D(F)$ will 
    decrease by at most $120\eps\Delta^2$ before we add a new 
    edge to the colorful matching. So each time when we add a new edge into the colorful matching, $a_D(F)$ decreases by at most $138\eps\Delta^2$ with high probability. 
    We now lower bound the value of $a_D(F)$ which allows us to bound the 
    number of times an edge is added to the colorful matching. 
    
    Let $e = (u,v)$ be an edge in $F$. Both $u$ and $v$ belong to the almost-clique $C$ in the extended HSS-decomposition and hence
    by Lemma~\ref{lem:extended-HSS-decomposition}, each have at most $7\eps\Delta$ neighbors outside $C$. This suggests that even if $\col_2$ has assigned
    a color to all vertices except for $C$, there are at least $(1-14\eps)\Delta$ available colors for the edge $e$, i.e., $a_D(e) \geq (1-14\eps)\Delta$. 
    Moreover, by Lemma~\ref{lem:extended-HSS-decomposition}, we also have that the number of vertices in the almost-clique $C$ and hence also in
    $\bar{C}$ is $N \geq (1-\eps)\Delta$. As such,
    \begin{align*}
        a_D(F)   &= \sum_{e\in F} a_D(e) \ge \card{F}(1-14\epsilon)\Delta \\
        &= (1-14\epsilon)\cdot \frac{\bar{d} N}{2}\cdot \Delta \ge \frac{1}{2}(1-14\epsilon)(1-\epsilon)\bar{d} \Delta \\
        &\ge 0.45 \bar{d} \Delta^2,
    \end{align*}
    by the choice of $\eps$. 
    Consequently, before $a_D(F)$ becomes smaller than $120\eps\Delta^2$ (and we could no longer apply Lemma~\ref{lem:colorvar}), 
    we would have added at least $\frac{0.45\bar{d}\Delta^2}{138\epsilon \Delta^2} \ge \frac{\bar{d}}{320\epsilon}$ edges to the colorful matching with high probability, finalizing the proof. 
\end{proof}

The coloring $\col_2$ is then computed as follows. We iterate over almost-cliques $C_1,\ldots,C_k$ and for each one, we pick a colorful matching of size $\frac{\bar{d}}{320\eps} \geq 4\bar{d}$ (by our choice of $\eps$); by Lemma~\ref{lem:colorful}, we
find this matching with high probability. We only pick $4\bar{d}$ edges from this colorful matching and for each edge $(u,v)$ in these $4\bar{d}$ edges, we set $\col_2(u) = \col_2(v) = c_{u,v}$. By Definition~\ref{def:colorful}, this is a
valid coloring. We then move to the next almost-clique (and use Lemma~\ref{lem:colorful} with the updated $\col_2$).

\subsection{Part Three: Final Coloring of Almost-Cliques.} We now finalize the coloring of almost-cliques and obtain a coloring $\col_3: V \rightarrow [\Delta+1] \cup \set{\perp}$ that assigns a valid color to all vertices. 
Initially, $\col_3(v) = \col_2(v)$ for all $v \in V$. We then iterate over almost-cliques $C_1,\ldots,C_k$ and update $\col_3$ to assign a valid color to all vertices of the current almost-clique. 
For any $C_i$, define $\hC_i$ as the vertices that are not yet assigned a valid color in $\col_3$. 

\begin{definition}[Palette-Graph]\label{def:palette-graph}
    For any almost-clique $C_i$ in $G$ and a partial coloring $\col_3$, we define a \emph{palette-graph} $H_i$ of the almost-clique with respect to $\col_3$ as follows: 
    \begin{itemize}
        \item $H_i$ is a bipartite graph between the uncolored vertices in $C_i$ (i.e., $\hC_i$) and colors $[\Delta+1]$.
        \item There exists an edge between $u \in \hC_i$ and $c \in [\Delta+1]$ iff the color $c$ is available to vertex $u$ according to $\col_3$ (i.e., $c \in A_{\col_3}(u)$)
         and moreover $c \in L_3(u)$. 
    \end{itemize}
\end{definition}

\noindent
Suppose we are able to find a matching in the palette-graph of an almost-clique $C_i$ that matches all vertices in $\hC_i$. Let $c_v$ be the matched pair of $v \in \hC_i$. 
We set $\col_3(v) = c_v$ and correctly color all vertices in this almost-clique, and then continue to the next almost-clique. 
The following lemma proves that with high probability, we can find such a matching in every almost-clique. 

\begin{lemma}\label{lem:palette-graph-matching}
    Let $C_i$ be any almost-clique in $G$ and $\col_3$ be the partial coloring obtained after processing almost-cliques $C_{1},\ldots,C_{i-1}$. With high probability (over the 
    randomness of the third batch), there exists 
    a matching in the palette-graph of $C_i$ that matches all vertices in $\hC_i$. 
\end{lemma}
\noindent
Proof of this lemma is based on the following result on existence of large matchings in certain random graphs 
that we prove in this paper. 

\begin{lemma} \label{lem:lamatching}
    Suppose $N\le n$ and $0 \le \delta \le \frac{1}{12}$. Let $H=(V_1,V_2)$ be a bipartite graph such that: 
    \begin{enumerate}
        \item \label{lam1} $\card{V_1} \le (1-3\delta)N$ and $\card{V_2} \le 2N$;
        \item \label{lam2} each vertex in $V_1$ has degree at least $\frac{2N}{3}$ and at most $N$;
        \item \label{lam3} the average degree of vertices in $V_1$ is at least $(1-\delta)N$. 
    \end{enumerate}
    A subgraph of $H$ obtained by sampling each edge with probability $p=\frac{90 \log n}{N}$ has a matching of size $|V_1|$ with high probability.
\end{lemma}

The proof of Lemma~\ref{lem:lamatching} is quiet technical and hence we postpone it to Section~\ref{sec:lamatching} to keep the flow of the current argument. 
We now use this result to prove Lemma~\ref{lem:palette-graph-matching}.

\begin{proof}[Proof of Lemma~\ref{lem:palette-graph-matching}]
    Define $H_i$ as the bipartite graph with the same vertex set as the palette-graph of $C_i$ such that there is an edge between a vertex $v$ and a color $c$ iff $c$ is 
    available to $v$ (edges in $H_i$ are superset of the ones in palette-graph as an edge $(v,c)$ can appear in $H_i$ even if $c \notin L_3(v)$). 
    By this definition, the palette-graph of $C_i$ is a subgraph of $H_i$ chosen by picking each edge independently 
    with probability $\frac{\alpha\log n}{3\eps^2\Delta} \geq \frac{100\log{n}}{\Delta}$ (by the choice of $L_3$). 

    We apply Lemma~\ref{lem:lamatching} to a properly chosen subgraph $\bar{H}$ of $H_i$ with the same vertex-set to prove this lemma. 
    Let $n_i$ be the number of vertices in $C_i$. By definition of coloring $\col_2$ (after the proof of Lemma~\ref{lem:colorful}), $\hC_i$ has $n_i-2\cdot 4\bar{d} = n_i-8\bar{d}$ vertices. For any vertex $v \in \hC_i$, if $v$ has more than $n_i-4\bar{d}$ 
    available colors (i.e., neighbors in $H_i$), then we pick $n_i-4\bar{d}$ available colors for $v$ arbitrarily and 
    only connect $v$ to those color-vertices in $\bar{H}$; otherwise, $v$ has the same neighbors in $\bar{H}$ and $H_i$. 
    
    We prove that $\bar{H}$ satisfies all three properties in Lemma \ref{lem:lamatching}. Let $N=n_i-4\bar{d}$ and $\delta = \frac{1.3\bar{d}}{N}$, $V_1= \hC_i$ and $V_2 := [\Delta+1]$, and thus $\card{V_1}=\card{\hC_i} = (n_i-8\bar{d}) < (1-3\delta)N$. This 
    proves the first part of Property~(\ref{lam1}) of Lemma~\ref{lem:lamatching}. Moreover, as $C_i$ is an almost-clique, $n_i \geq (1-\eps)\Delta$ by Property~(\ref{HSS-p2a}) and $\bar{d} \leq 6\eps\Delta$ by 
    Property~(\ref{HSS-p2c}) of Lemma~\ref{lem:extended-HSS-decomposition}, and hence $2N=2(n_i-4\bar{d})\ge 2((1-\epsilon)\Delta - 4 \cdot 6\epsilon\Delta) > \Delta+1=\card{V_2}$, proving the second part 
    of Property~(\ref{lam2}) as well. 
    
     Furthermore, each vertex in $\hC_i$ has degree at most $n_i-4\bar{d}=N$. Also any vertex in $\hC_i$ has degree at most $7\eps \Delta$ outside the almost-clique in $G$ by Property (\ref{HSS-p2b}) of Lemma \ref{lem:extended-HSS-decomposition}, so any vertex in $\hC_i$ has at least $(1-7\eps)\Delta-4\bar{d}>\frac{2N}{3}$ available colors (even if $\col_3$ has assigned colors to all vertices outside $C_i$ and all colors used by the colorful matching are also unavailable). As in $\bar{H}$ we connect every vertex to the available color-vertices, $\bar{H}$ satisfies Property~(\ref{lam2}) in Lemma \ref{lem:lamatching}.
    
    Consider a vertex $v \in \hC_i$. Let $\bar{d}_v$ be the number of non-neighbors 
    of $v$ inside $C_i$ and hence $v$ has at most $\Delta-(n_i-\bar{d}_v-1)$ neighbors outside $C_i$. As such, $v$ has at least $n_i-\bar{d}_v-4\bar{d}$ neighbors inside of $H_i$ ($4\bar{d}$ is the number of colors used by the colorful matching), hence also
    at least $n_i-\bar{d}_v-4\bar{d}$ neighbors inside of $\bar{H}$. So $\bar{H}$ has at least $\card{\hC_i}(n_i-4\bar{d})-\bar{d}n_i \ge \card{\hC_i}(n_i-5.2\bar{d})$ edges (by the fact that $\card{\hC_i} = n_i-8\bar{d} \ge n_i-48\eps \Delta$ and the choice of $\eps$). 
    Hence, the average degree in $\bar{H}$ is at least $n_i-5.2\bar{d} = N - 1.2\bar{d} \geq (1-\delta)N$, which implies $\bar{H}$ satisfies Property~(\ref{lam3}) in Lemma \ref{lem:lamatching}. 
    
   To conclude, $\bar{H}$ satisfies the properties of Lemma~\ref{lem:lamatching}. Since the palette-graph of $C_i$ contains a subgraph of $\bar{H}$ obtained by sampling each edge of $\bar{H}$ with probability $\frac{100\log n}{\Delta} \ge \frac{90 \log n}{N}$ (as argued above), the palette-graph contains a matching of size $\card{V_1} = \card{\hC_i}$ with high probability. 
\end{proof}

\subsection{Wrap-Up: Proof of Theorem~\ref{thm:color-sampling}.} The existence of list-coloring under the sampling process of $L_1,L_2$ and $L_3$ follows
from Lemmas~\ref{lem:sparse-color},~\ref{lem:colorful} and~\ref{lem:palette-graph-matching}. Note that this sampling process is not exactly the same as sampling $O(\log{n})$ colors uniformly at random as in Theorem~\ref{thm:color-sampling}. 
However, in this process, with high probability, we do not sample more than $O(\log{n})$ colors for each vertex and hence 
conditioning on sampling $O(\log{n})$ colors (as in Theorem~\ref{thm:color-sampling}) only changes the probability of success by a negligible factor, hence implying Theorem~\ref{thm:color-sampling}. 

\subsection{Proof of Lemma~\ref{lem:lamatching}: Large Matchings in (Almost) Random Graphs}\label{sec:lamatching}

We first need the following auxiliary claim. The proof is standard and appears in Appendix~\ref{app:sumvar}. 

\begin{claim} \label{lem:sumvar}
    Suppose $0<a<1$ is a constant. Consider two random variables $X:=\sum_{i=1}^{n}X_i$ and $Y:=\sum_{i=1}^{n}Y_i$ where for all $i \in [n]$, $X_i$ and $Y_i$ are independent indicator random variables and $\prob{X_i=1}=a^{k_i}$ and $\prob{Y_i=1}=a^{\ell_i}$. Suppose $k_i$ and $\ell_i$ are non-negative integers that are indexed in decreasing order and have the following two properties:
    \begin{itemize}
        \item for any $j \le n$, $\sum_{i=1}^j k_i \le \sum_{i=1}^j \ell_i$
        \item $\sum_{i=1}^n k_i = \sum_{i=1}^n \ell_i$
    \end{itemize}
    then for any integer $M$, $\prob{X\ge M}\le \prob{Y\ge M}$.
\end{claim}

We now use this claim to prove Lemma~\ref{lem:lamatching} restated below. 

\begin{lemma*}[Restatement of Lemma~\ref{lem:lamatching}]
    Suppose $N\le n$ and $0 \le \delta \le \frac{1}{12}$. Let $H=(V_1,V_2)$ be a bipartite graph such that: 
    \begin{enumerate}
        \item \label{lam1a} $\card{V_1} \le (1-3\delta)N$ and $\card{V_2} \le 2N$;
        \item \label{lam2a} each vertex in $V_1$ has degree at least $\frac{2N}{3}$ and at most $N$;
        \item \label{lam3a} the average degree of vertices in $V_1$ is at least $(1-\delta)N$. 
    \end{enumerate}
    If $H'$ is a subgraph of $H$ obtained by sampling each edge with probability $p=\frac{90 \log n}{N}$, then $H'$ has a matching of size $|V_1|$ with high probability.
\end{lemma*}

\begin{proof}[Proof of Lemma~\ref{lem:lamatching}]
    By Hall's marriage theorem, we only need to prove that with high probability, for any $S_1 \subseteq V_1$, the size of neighbor set ${N(S_1)}$ of $S_1$ is at least $|S_1|$ in $H'$, i.e., $\card{N(S_1)} \geq \card{S_1}$.
    Fix a sets $S_1$ and let $k:= \card{S_1}$; we prove this for the set $S_1$.

    If $k\le \frac{n}{3}$, since each vertex has degree at least $\frac{2N}{3}$ in $H$, total number of edges from $S_1$ to $S_2$ is at least $\frac{2kN}{3}$. On the other hand, if we fix another set $S_2 \subseteq V_2$ with $|S_2|=|V_2|+1-k$, there are at most $k^2$ edges between $S_1$ and $V_2 \setminus S_2$ due to the fact that both $S_1$ and $V_2 \setminus S_2$ have at most $k$ vertices. As such, the number of edges between $S_1$ and $S_2$ in $H$ is at least $\frac{2kN}{3}-k^2 \ge \frac{2kN}{3}-\frac{kN}{3} \ge \frac{kN}{3}$. The probability that there is no edge between $S_1$ and $S_2$ in $H'$ is at most $(1-p)^{\frac{kN}{3}} \le n^{-30k}$. Taking the union bound over every subset $S_2$ of size $|V_2|+1-k$ (there are at most ${{\card{V_2}}\choose{k}} \le {{2N}\choose{k}} \le N^{2k} \le n^{2k}$ of such sets), the probability that $|N(S_1)|\ge k$ is at least $1-\frac{1}{n^{20k}}$.

    If $k \ge \frac{n}{3}$, the number of edges incident on $V_1 \setminus S_1$ is at most $(\card{V_1}-k)N$. Hence, the number of edges incident on $S_1$ is at least $\card{V_1}((1-\delta)N) - (\card{V_1}-k)N=kN - \delta |V_1| N$ (the total number of edges in the graph is at least $\card{V_1}(1-\delta)N$). Additionally, since $k \ge \frac{n}{3} \ge \frac{N}{3}$, we have $\delta \card{V_1} N \le 3\delta k \card{V_1}$. The number of edges incident on $S_1$ is at least $k(N-3\delta \card{V_1}) \ge k(N-3\delta\card{N}) = kN(1-3\delta))$. By the fact that $\card{V_1}\le (1-3\delta)N$, the number of edges incident on $S_1$ is at least $k\card{V_1}$.
  
    For each vertex $v_i\in V_2$, let $k_i$ be the number of edges from $S_1$ to $v_i$ in $H$ and $X_i$ be the indicator random variable which is $1$ iff $v_i$ is \emph{not} a neighbor of $S_1$ in $H'$. The number of non-neighbors of $S_1$ in $H'$ is $X=X_1+X_2+\dots X_{\card{V_2}}$ where $\prob{X_i=1}=(1-p)^{k_i}$. If $X\le\card{V_2}-k$, then the number of neighbors of $S_1$ in $H'$ is at least $k$ and we are done. As such, we need to bound $\prob{X>\card{V_2}-k}$.

    To prove that $\prob{X>\card{V_2}-k}$ is low, we will first define another random variable $Y$, and use Claim \ref{lem:sumvar} to prove that $\prob{X>\card{V_2}-k}\le \prob{Y>\card{V_2}-k}$, and then prove that $\prob{Y>\card{V_2}-k}$ is low. 
    First, we pick $\card{V_2}$ non-negative integers $k'_1 \le k_1, k'_2 \le k_2, \dots, k'_{\card{V_2}}\le k_{\card{V_2}}$, so that $\sum_{i=1}^{\card{V_2}} k'_i = k \card{V_1}$ (we can do so because the sum of $k_i$ is the number of edges in $H$ incident on $S_1$, which is larger than $k\card{V_1}$). Define a random variable $X'=X'_1+X'_2+\dots+X'_{\card{V_2}}$ where $X'_i$ are independent indicator random variables such that $\prob{X'_i=1}=(1-p)^{k'_i}$. As $\prob{X_i = 1} \leq \prob{X'_i=1}$ for all $i$, by a coupling argument, we have $\prob{X>\card{V_2}-k}\le \prob{X'>\card{V_2}-k}$. 

    Now we will define a random variable $Y$ such that $X'$ and $Y$ satify the condition of Claim \ref{lem:sumvar}. Without loss of generation, suppose $k'_i$ is in decreasing order. Define $\ell_1,\ell_2,\dots,\ell_{\card{V_2}}$ as follows: if $1 \le i \le \card{V_1}$, then $\ell_i=k$, otherwise $\ell_i=0$. Let $Y=Y_1+Y_2+\dots+Y_{\card{V_2}}$ where $Y_i$ are independent indicator random variables such that $\prob{Y_i=1}=(1-p)^{\ell_i}$. We have $\sum_{i=1}^{\card{V_2}} \ell_i=k\card{V_1}=\sum_{i=1}^{\card{V_2}}k'_i$. So $X'$ and $Y$ satisfy the second condition of Claim \ref{lem:sumvar}. Furthermore, for any $j$ with $1 \le j \le \card{V_2}$, if $j>\card{V_1}$, then $\sum_{i=1}^j \ell_i = k\card{V_1} \ge \sum_{i=1}^j k'_i$; otherwise if $j \le \card{V_1}$, then $\sum_{i=1}^j \ell_i = k \cdot j \ge \sum_{i=1}^j k'_i$ because $k'_i\le k_i \le k$ for any $i$. So $X'$ and $Y$ satisfy the first condition of Claim \ref{lem:sumvar}. By applying Claim~\ref{lem:sumvar}, we have $\prob{X'>\card{V_2}-k} \le \prob{Y>\card{V_2}-k}$.
    
    Now we will prove that $\prob{Y>\card{V_2}-k}$ is low. Note that for any $i>\card{V_1}$, $Y_i$ is always $1$, so $Y=\sum_{i=1}^{\card{V_1}} Y_i + \card{V_2} - \card{V_1}$. Hence we need to prove that $\prob{\sum_{i=1}^{\card{V_1}} Y_i >\card{V_1}-k}$ is low.
    For any subset $S_2$ of $\{Y_i \mid i \le |V_1|\}$ with size $|V_1|+1-k$, the probability that they are all $1$ is $(1-p)^{k(|V_1|+1-k)} \le n^{-30(|V_1|+1-k)}$. By taking a union bound over all choices of $S_2$ (the number of such sets is at most ${{\card{V_1}}\choose{\card{V_1}+1-k}} \le n^{\card{V_1}+1-k}$), $\prob{Y > \card{V_2}-k} = \prob{\sum_{i=1}^{\card{V_1}} Y_i > \card{V_1}-k} \le n^{-20(|V_1|+1-k)}$.

    Finally, for any $k$, by taking a union bound on all possible choices for $S_1$ with size $k$, the probability that $N(S_1)<k$ is at most $n^{-15}$ since the number of such $S_1$ is at most ${\card{V_1}}\choose{k}$, which is at most $n^k$ and at most 
    $n^{\card{V_1}-k}$. A union bound on $n$ values of $k$ now finalizes the proof, as with probability at least $1-1/n^{14}$, no set $S_1$ has $\card{S_1} > \card{N(S_1)}$.
\end{proof}

\renewcommand{\CC}{\ensuremath{\chi}}

\newcommand{\CMA}{\ensuremath{\textnormal{\textsf{ColoringAlgorithm}}}\xspace}

\newcommand{\Ec}{\ensuremath{E_{\textnormal{\textsf{conflict}}}}}
\newcommand{\Gc}{\ensuremath{G_{\textnormal{\textsf{conflict}}}}}

\section{Sublinear Algorithms for $(\Delta+1)$ Coloring}\label{sec:sublinear}

We now use our palette-sparsification theorem to design sublinear algorithms for $(\Delta+1)$ coloring in different models of computation. We start by introducing a ``meta-algorithm'' for $(\Delta+1)$ coloring, called \CMA, and in the subsequent sections show how to implement
it in each model. 

\paragraph{The Meta-Algorithm.} The algorithm is as follows: 

\begin{tbox}
$\CMA(G,\Delta)$: A meta-algorithm for finding a $(\Delta+1)$-coloring in a graph $G(V,E)$ with maximum degree $\Delta$.

\begin{enumerate}[leftmargin=15pt]
\item Sample $\Theta(\log{n})$ colors $L(v)$ uniformly at random for each vertex $v \in V$ (as in Theorem~\ref{thm:color-sampling}). 
\item Define, for each color $c \in [\Delta+1]$, a set $\CC_c \subseteq V$ where $v \in \CC_c$ iff $c \in L(v)$. 
\item Define $\Ec$ as the set of all edges $(u,v)$ where both $u,v \in \CC_c$ for some $c \in [\Delta+1]$. 
\item {\underline{Construct}} the \emph{conflict graph} $\Gc(V,\Ec)$. 
\item \underline{Find} a proper list-coloring of $\Gc(V,\Ec)$ with $L(v)$ being the color list of vertex $v \in V$.  
\end{enumerate} 
\end{tbox}

 We refer to $\CMA$ as a ``meta-algorithm'' since constructing the conflict graph as well as finding its list-coloring are \emph{unspecified} steps in \CMA. 
To implement this meta-algorithm in different models, we need to come up with an efficient way of performing these two tasks which are model-specific and are hence not fixed in \CMA. 
The following lemma establishes the main properties of $\CMA$. 

\begin{lemma}\label{lem:meta-algorithm}
	Let $G(V,E)$ be a graph with maximum degree $\Delta$. In $\CMA(G,\Delta)$, with high probability: 
	\begin{enumerate}
		\item The output is a valid $(\Delta+1)$ coloring of the graph $G$. 
		\item For any $c \in [\Delta+1]$, size of $\CC_c$ is $O(n\log{n}/\Delta)$.
		\item The maximum degree in graph $\Gc$ is $O(\log^2{n})$. 
	\end{enumerate}
\end{lemma}
\begin{proof}
	We show that each part holds with high probability. Taking a union bound over the three parts finalizes the proof. 
	\begin{enumerate}
	\item We apply Theorem~\ref{thm:color-sampling} to the sets $L(v)$ chosen for each $v \in V$, obtaining that with high probability, $G$ can be list-colored with $L(v)$ being the list of vertex $v$. Now notice 
	that since $\Ec$ contains all possible monochromatic edges that arise in any list-coloring of $G$ with lists $L(\cdot)$, any proper list-coloring of $G$ is a proper list-coloring of $\Gc$ and vice versa. As such, we know that
	$\Gc$ contains a proper list-coloring and this list-coloring is also a feasible $(\Delta+1)$ coloring of the graph $G$. 
	
	\item  Fix any color $c \in [\Delta+1]$. Let $K$ be the  number of colors sampled by each vertex. The probability that any specific vertex $v$ chooses $c$ in $L(v)$ is $K/(\Delta+1)$. As such, the expected
	number of vertices in $\CC_c$ is $n\cdot K/(\Delta+1)$. As $K = \Theta(\log{n})$ and the choice of $L(\cdot)$ is independent across all vertices, by Chernoff bound, the total number of vertices in $\CC_c$ is with high probability 
	$2n \cdot K/\Delta = O(n\log{n}/\Delta)$. Taking a union bound on all $\Delta+1$ classes, finalizes the proof of this part. 
	
	\item Fix any vertex $v \in V$ and again let $K$ be the number of colors sampled by each vertex. 
	We fix these colors, say, $c_1,\ldots,c_K$. For any neighbor of $v$, say $u \in N(v)$ and $i \in [K]$, let $X_{u,i}$ be an indicator random variable which is one iff $c_i \in L(u)$. 
	Let $X := \sum_{u \in N(v)} \sum_{i=1}^{K} X_{u,i}$ and thus $\Ex\bracket{X} = \Delta \cdot K \cdot K/(\Delta+1) \leq K^2$. Note that $X$ is an upper bound on degree of $v$ in $\Gc$. As the $X_{u,i}$'s are negatively correlated, by Chernoff bound, we 
	have that $X \leq O(K^2) = O(\log^{2}{n})$ with high probability. Taking a union bound on all $n$ vertices finalizes the proof of this part. 
	\end{enumerate}
\noindent	
	Lemma~\ref{lem:meta-algorithm} now follows from a union bound over the three events above. 
\end{proof}

In the remainder of this section, we use Lemma~\ref{lem:meta-algorithm} to design sublinear algorithms for $(\Delta+1)$ coloring in graph streaming, sublinear time, and massively parallel computation (MPC) models.

\newcommand{\pf}{\ensuremath{\textnormal{\textsf{DetectPotentialNeighbor}}}\xspace}
\newcommand{\pPF}{\ensuremath{p}}
\newcommand{\KPF}{\ensuremath{K}}

\subsection{A Single-Pass Streaming Algorithm}\label{sec:streaming}

We first give an application of our palette-sparsification theorem in designing a dynamic streaming algorithm for the $(\Delta+1)$ coloring problem. In the dynamic streaming model, 
the input graph is presented as an arbitrary sequence of edge insertions and deletions and the goal is to analyze properties of the resulting graph using memory that is sublinear in the input size, which is proportional to the number
of edges in the graph. We are particularly interested in algorithms that use $O(n \cdot \polylog{(n)})$ space, referred to as \emph{semi-streaming} algorithms~\cite{FKMSZ05}. 

\begin{theorem}\label{thm:streaming}
    There exists a randomized {single-pass} semi-streaming algorithm that given a graph $G$ with maximum degree $\Delta$ presented in a dynamic stream, with high probability finds a $(\Delta+1)$ coloring of $G$ with 
   using $\Ot(n)$ space and polynomial time\footnote{Here $\Delta$ is the maximum degree of the graph at the end of the stream and we assume no upper bound on degree of vertices throughout the stream, which can be as large
    as $\Omega(n)$ even when $\Delta$ is much smaller.}. 
\end{theorem}

We prove Theorem~\ref{thm:streaming} by implementing $\CMA$ in dynamic streams. Recall that implementing $\CMA$ requires us to specify $(i)$ how we construct the conflict-graph, and $(ii)$ how we find a list-coloring in this conflict graph using the lists 
$L(\cdot)$. Throughout the proof, we condition on the high probability event in Lemma~\ref{lem:meta-algorithm}. We first show how to construct the conflict-graph. To do so, we rely on the by now standard primitive of $\ell_0$-samplers for sampling elements
in dynamic streams (see,~e.g.~\cite{FrahlingIS2008,JowhariST2011,KapralovNPWWY17}) captured in the following proposition. 

\begin{proposition}[cf.~\cite{JowhariST2011,McGregorTVV15}]\label{prop:ell-0-samplers}
	There exists an streaming algorithm that given a subset $P \subseteq V \times V$ of pairs of vertices and an integer $k \geq 1$ at the beginning of a dynamic stream, outputs with high probability a set $S$ of $k$ edges 
	from the edges in $P$ that appear in the final graph (it outputs all edges if their number is smaller than $k$). The set $S$ of edges can be either chosen uniformly at random with replacement or without replacement.
	The space of algorithm is $O(k \cdot \log^{3}{n})$. 
\end{proposition}

Using Proposition~\ref{prop:ell-0-samplers}, we show how to construct the conflict graph in the  streaming model. 

\begin{lemma}\label{lem:streaming-construct}
    $\Gc(V,\Ec)$ can be constructed in $\Ot(n)$ space and polynomial time in dynamic streams with high probability. 
\end{lemma}
\begin{proof}
    We construct the sets $\CC_1,\ldots,\CC_{\Delta+1}$ and store the sets in $O(n\log{n})$ space. For any vertex $v \in V$, we define the set $P_v$ of all edge slots between $v$ 
    and $\bigcup_{c \in L(v)} \CC_c$, i.e., all vertices that may have an edge to $v$ in $\Ec$, and run the algorithm in Proposition~\ref{prop:ell-0-samplers} with $P = P_v$ and parameter $k = O(\log^2{(n)})$. 

    As we conditioned on the event in Lemma~\ref{lem:meta-algorithm}, the degree of each vertex is at most $k$ in $\Gc$. Hence, by Proposition~\ref{prop:ell-0-samplers}, with high probability,
    we find all neighbors of this vertex in $\Gc$. Taking a union bound over all vertices in $V$, with high probability, we can construct the graph $\Gc$. As all these steps can be implemented in polynomial time, we obtain the final result. 
\end{proof}

Lemma~\ref{lem:streaming-construct} together with Lemma~\ref{lem:meta-algorithm} are already enough to achieve a semi-streaming algorithm with exponential-time (i.e., prove Theorem~\ref{thm:streaming} if we do not want a polynomial time algorithm): after constructing
the conflict-graph $\Gc$, we can simply use an exponential time algorithm to find a proper list-coloring of $\Gc$ which would be a $(\Delta+1)$ coloring of $G$ by Lemma~\ref{lem:meta-algorithm}. 
We now show how to find a proper list-coloring of $\Gc$ in polynomial time. 

\paragraph{A Polynomial Time Streaming Algorithm for $(\Delta+1)$ Coloring.} Recall that the proof of Theorem~\ref{thm:color-sampling} was in fact constructive modulo one part: we assumed the existence of the extended HSS-decomposition and found the $
(\Delta+1)$ coloring using this assumption. However, without no explicit access to the extended HSS-decomposition, it is not clear how to find such a coloring. We remedy this by presenting a streaming algorithm that can compute this decomposition explicitly
in a small space in one pass over the stream. We then use the argument in Theorem~\ref{thm:color-sampling} to color the graph $G$ in polynomial time. 

\begin{lemma}\label{lem:find-decomposition}
    There exists a single-pass dynamic streaming algorithm that can find the extended HSS-decomposition of any given graph $G(V,E)$ with parameter $\eps > 0$ in 
    $\Ot(n/\eps^2)$ space and polynomial time with high probability. 
\end{lemma}

The rest of this section is devoted to the proof of Lemma~\ref{lem:find-decomposition}. Let $\eps$ be the parameter in Lemma~\ref{lem:find-decomposition} and define $\delta := \eps/10$ throughout this section. 
We prove Lemma~\ref{lem:find-decomposition} in multiple steps. The first step is to determine for all vertices $u,v \in V$, whether these vertices can be $\delta$-friend or not (as in Definition~\ref{def:eps-friend}).
Formally, we say that $u,v \in V$ are a \emph{$\delta$-potential friend} iff $\card{N(v) \cap N(u)} \geq (1-\delta)\Delta$. Note that the only difference between this definition and definition of $\delta$-friends in Definition~\ref{def:eps-friend} is 
that here $u$ and $v$ may \emph{not} be neighbor to each other. We design a simple data structure to determine (approximately) whether $u,v \in V$ are $\delta$-potential friend or not in a single pass over the dynamic stream. 

\begin{lemma}\label{lem:detect-potential-friend}
    There exists a single-pass dynamic streaming algorithm that constructs a data structure to answer the following queries: Given a pair of vertices $(u,v) \in V \times V$,
    \begin{itemize}
        \item outputs \Yes if $u$ and $v$ are $\delta$-potential friend;
        \item outputs \No if $u$ and $v$ are not $2\delta$-potential friend;
        \item the answer can be arbitrary otherwise. 
    \end{itemize}
    The algorithm requires $\Ot(n/\delta^2)$ space and outputs the correct data structure with high probability. 
\end{lemma}
\begin{proof}
    The algorithm is simply as follows: Pick a set $S$ of vertices at the beginning of the stream by choosing each vertex independently and with probability $p := \frac{10\log{n}}{\delta^2\Delta}$. For any chosen vertex in $S$, 
    run the algorithm in Proposition~\ref{prop:ell-0-samplers} with $P$ being the set of all edge slots incident to the vertex and $k = \Delta$. As the degree of every vertex is at most $\Delta$, 
    with high probability we recover all edges incident on $S$ at the end of the stream. In the following, we condition on this event. 

    To answer the query $(u,v)$ we do as follows: let $c_{u,v}$ be the number of common neighbors of $u$ and $v$ in $S$ (as we have all the edges incident on $S$ we can determine this quantity). Now
    output $\Yes$ iff $c_{u,v} \geq (1-1.5\delta)\cdot \Delta \cdot p$ and $\No$ otherwise. Let $N(u,v)$ denote the set of common neighbors of $u$ and $v$. The proof of correctness is as follows: 
    \begin{itemize}
        \item Suppose $u$ and $v$ are $\delta$-potential friend and hence $\card{N(u,v)} \geq (1-\delta)\Delta$. In expectation, $\Ex[c_{u,v}] = \card{S \cap N(u,v)} \geq (1-\delta)\cdot\Delta \cdot p$ and hence by Chernoff bound, with probability
            at least $1-1/n^{5}$, $c_{u,v} \geq (1-1.5\delta)\cdot \Delta \cdot p$. 

        \item Suppose $u$ and $v$ are not $2\delta$-potential friend and hence $\card{N(u,v)} < (1-2\delta)\Delta$. The above argument implies that $c_{u,v} < (1-1.5\delta)\cdot\Delta \cdot p$ with probability $1-1/n^5$.  
    \end{itemize}

    By taking a union bound over all $n^2$ $(u,v)$-pairs, the answer of the data structure is correct with high probability for all pairs. 

    Finally, another direct application of Chernoff bound ensures that the total number of vertices chosen in $S$ is at most $O(n\log{n}/\Delta)$ with high probability. As each vertex requires $O(\Delta \cdot \polylog{(n)})$ space to  
    run the algorithm in Proposition~\ref{prop:ell-0-samplers}, we obtain the bound on the space requirement of the algorithm as well. 
\end{proof}

For the rest of the proof, we pick the data structure in Lemma~\ref{lem:detect-potential-friend} in our algorithm and assume that the high probability event in Lemma~\ref{lem:detect-potential-friend} happens. 
We now use this data structure to determine (approximately) if a vertex is $\delta$-dense or not. 

\begin{lemma}\label{lem:detect-eps-dense}
    There exists a single-pass dynamic streaming algorithm that uses $\Ot(n/\delta^2)$ space with high probability outputs a set $D$ of vertices such that: 
    \begin{itemize}
        \item any vertex in $D$ is at least $2\delta$-dense;
        \item all $(\delta/2)$-dense vertices belong to $D$. 
    \end{itemize}
\end{lemma}
\begin{proof}
    We sample $k:= \frac{100 \cdot n\log{n}}{\delta^2}$ edges $E_s$ from the graph uniformly at random with replacement using Proposition~\ref{prop:ell-0-samplers}. At the end of the stream, 
    for any edge $(u,v) \in E_s$ we use the data structure in Lemma~\ref{lem:detect-potential-friend} on the pair $(u,v)$ and discard this edge iff the answer to the query is \No. Let $E'_s$ be the set of non-discarded edges. 
    We add any vertex $u$ to $D$ iff there are $ \geq (1-\delta) \cdot \frac{k \cdot \Delta}{n \cdot \Delta}$ edges in $E'_s$ incident on $u$. 

    Suppose first that $u$ is a $(\delta/2)$-dense vertex and let $F_{u}$ denote the set of $(\delta/2)$-friend edges incident on $u$. By Definition~\ref{def:eps-dense}, $\card{F_u} \geq (1-(\delta/2))\Delta$. Moreover, by the guarantee of 
    the data structure in Lemma~\ref{lem:detect-potential-friend}, for any edge $(u,v) \in F_u$, the answer to the query $(u,v)$ is \Yes, as $(u,v)$ is a $(\delta/2)$-friend edge (and hence is also $(\delta/2)$-potential friend by definition). 
    Now let $c_{u} := \card{F_u \cap E_s}$ which is also a lower bound on the number of neighbors of $u$ in $E'_s$. By the choice of $E_s$, we have that $\Ex\bracket{c_{u}} \geq (1-(\delta/2)) \cdot \frac{k \cdot \Delta}{n \cdot \Delta}$. 
    An application of Chernoff bound now ensures that $c_{u} \geq (1-\delta) \cdot \frac{k \cdot \Delta}{n \cdot \Delta}$ with probability at least $1-1/n^5$.  Hence, $u$ would be added to $D$ with probability at least $1-1/n^5$. 

    Now suppose $u$ is not a $2\delta$-dense vertex. Let $F_u$ denote the set of $(2\delta)$-friend edges incident on $u$. By Definition~\ref{def:eps-dense}, $\card{F_u} \leq (1-2\delta)\Delta$. By the guarantee of 
    the data structure in Lemma~\ref{lem:detect-potential-friend}, for any edge $(u,v) \notin F_u$, the answer to the query $(u,v)$ in \No as $(u,v)$ in that case would not be a $(2\delta)$-potential friend. Let $c_u := \card{F_u \cap E_s}$
    which is an upper bound on the number of neighbors of $u$ in $E'_s$. By the choice of $E_s$, we have $\Ex\bracket{c_u} \leq (1-2\delta) \frac{k \cdot \Delta}{n \cdot \Delta}$. As such, a Chernoff bound ensures that
    with probability  at least $1-1/n^5$, $u$ would not be added to $D$. 

    Taking a union bound over all $n$ vertices finalizes the proof of correctness. The space needed by the algorithm is also $\Ot(n/\delta^2)$ by Proposition~\ref{prop:ell-0-samplers}.  
\end{proof}

We are now ready to construct the extended HSS-decomposition. Let $D$ be the set of vertices found in Lemma~\ref{lem:detect-eps-dense} and for the rest of the proof, we assume that the high probability event 
in Lemma~\ref{lem:detect-eps-dense} happens. Define the following graph $H$ with the set of vertices $V(H) := D$ and the set of edges:
\begin{align*}
    E(H):= &\{(u,v) \in E \mid \text{data structure in Lemma~\ref{lem:detect-potential-friend}} 
    \text{~returns \Yes to the query $(u,v)$}\}.
\end{align*}
Notice that we cannot construct the graph $H$ explicitly and we only use it in the analysis. However, we can construt a graph $H_s$ which is an edge-sampled subgraph of $H$. We will prove that $H_s$ will have the same decomposition as $H$. 
Before showing how to construct $H_s$, we first recall a result about the connectivity of a subgraph obtained by sampling each edge of an original dense graph. 

\begin{lemma} [\!\!\cite{alon1995note}] \label{lem:cut-connect}
    Let $G_p$ be a subgraph of a graph $G$ chosen by sampling each edge of $G$ with probability $p$. For any constant $b>0$, there exists a constant $c=c(b)$, such that if the minimum cut in $G$ is of size at least $\paren{c \log n/p}$, then
    $G_p$ is connected with probability at least $1-1/n^b$.
\end{lemma}

Let $E_s$ be a collection of $\ell:= \frac{10c \cdot n\log{n}}{\delta^2}$ (where $c$ is the constant $c(b)$ in Lemma \ref{lem:cut-connect} with $b=3$) edges chosen uniformly at random without replacement from the graph, which we pick 
in $\Ot(n/\delta^2)$ space using Proposition~\ref{prop:ell-0-samplers}. Define the graph $H_s$ with the set of vertices $V(H_s) := D$ and the set of edges: 
\begin{align*}
    E(H_s) := &\{(u,v) \in E_s \mid \text{data structure in Lemma~\ref{lem:detect-potential-friend}} \text{returns \Yes to the query $(u,v)$}\}.
\end{align*}
We can construct the graph $H_s$ explicitly at the end of the stream. We note that an important property of $E(H_s)$ is that any edge in $E(H)$ which is sampled in $E_s$ also appears in $E(H_s)$, i.e., $E(H) \cap E_s \subseteq E(H_s)$; this 
is by the guarantee of  Lemma~\ref{lem:detect-potential-friend}. 

We first prove that the connected components of $H$ has similar properties to almost-cliques in the extended HSS-decomposition. 
The following lemma is similar to Lemma \ref{lem:HSS-decomposition}.

\begin{lemma} \label{lem:app-HSS-decomposition}
    With high probability for any connected component $C_i$ in $H$: 
    \begin{enumerate}
   	\item $\card{C_i} \leq (1+6\delta)\Delta$;
	\item for any vertex $v$ in $C_i$, we have at most $6\delta\Delta$ non-neighbors (in $G$) inside of $C_i$.
   \end{enumerate}
\end{lemma}
\begin{proof}
    Let $H_{2\delta}$ be the subgraph of $G$ that only contains $(2\delta)$-dense vertices and $(2\delta)$-friend edges. As we conditioned on the events in Lemma \ref{lem:detect-potential-friend} and Lemma \ref{lem:detect-eps-dense}, by definition, 
    $H \subseteq H_{2\delta}$. As such, any connected component $C_i$ in $H$ is a subset of a connected component $C'_i$ in $H_{2\delta}$. By Lemma \ref{lem:HSS-decomposition}, $C'_i$ has at most $(1+6\delta)$ vertices, and for any vertex $v \in C'_i$, there are at most $6\delta\Delta$ non-neighbors (in $G$) inside of $C'_i$. These two properties are also held by $C_i$ since $C_i$ is a subset of $C'_i$. 
\end{proof}
\noindent
The following lemma is similar to Lemma \ref{lem:extended-HSS-decomposition} 

\begin{lemma} \label{lem:appext-HSS-decomposition}
    With high probability, any connected component $C_i$ in $H$ which has a $(\delta/4)$-dense vertex has at least $(1-\delta/4)\Delta$ vertices.
\end{lemma}

\begin{proof}
    By the proof of Lemma \ref{lem:extended-HSS-decomposition}, a $(\delta/4)$-dense vertex $v$ has at least $(1-\delta/4)\Delta$ neighbors which are $(\delta/2)$-dense and $(\delta/2)$-friend with $v$. As we conditioned on the events
    in Lemma \ref{lem:detect-potential-friend} and Lemma \ref{lem:detect-eps-dense}, these vertices as well as their edges to $v$ appear in the graph $H$. This implies that  
    the connected component $C_i$ which contains $v$ has at least $(1-\delta/4)\Delta$ vertices. 
\end{proof}

By Lemma \ref{lem:appext-HSS-decomposition}, we can see that any connected component of $H$ which has a $(\delta/4)$-dense vertex has the same properties as the components of the extended HSS-decomposition in Lemma \ref{lem:extended-HSS-decomposition} (for some $\eps = \Theta(\delta)$). Recall that however our streaming algorithm computes the graph $H_s$ not $H$. In the following, we prove that this is not a problem as the connected components of $H_s$ and $H$ that contain 
a $(\delta/4)$-dense vertex are the same with high probability. 

\begin{lemma} \label{lem:Hs-connect}
    With high probability, any connected component $C_i$ in $H$ which has a $(\delta/4)$-dense vertex is also connected in $H_s$.
\end{lemma}

\begin{proof} 
    Fix any connected component $C_i$ of $H$ which contains a $(\delta/4)$-dense vertex. In the following, we condition on the events in Lemmas~\ref{lem:app-HSS-decomposition} and~\ref{lem:appext-HSS-decomposition}. 
    By Lemma \ref{lem:appext-HSS-decomposition}, $\card{C_i} \ge (1-\delta/4)\Delta \geq (1-\delta)\Delta$.  Moreover, by Lemma~\ref{lem:app-HSS-decomposition}, for any vertex $v\in C_i$, $v$ has at most $6\delta\Delta$ non-neighbors inside of $C_i$. So any 
    vertex $v\in C_i$ has at least $(1-7\delta)\Delta$ neighbors inside of $C_i$. 
    
    Fix any cut $(V_1,V_2)$ of $C_i$ with $\card{V_1} \le \card{V_2}$. We first have $\card{V_1} \le (1+6\delta)\Delta/2$ by the upper bound on size of $C_i$ in Lemma~\ref{lem:app-HSS-decomposition}. Moreover, any vertex $v \in V_1$ has at least 	
    $(1-7\delta)\Delta-\card{V_1}$ neighbors in $V_2$ by the lower bound on the number of neighbors of $v$ in $C_i$, which means the size of the cut $(V_1,V_2)$ is at least $\card{V_1}((1-7\delta)\Delta-\card{V_1})$. By the choice of $\delta$, we have 
    $(1-7\delta)\Delta-\card{V_1}>0$, so the cut size is at least $(1-7\delta)\Delta-1 \ge 0.5\Delta$ with high probability. 
    
    Let $c$ be the constant in Lemma \ref{lem:cut-connect} when $b$ is 3. Define the subgraph $H'_s$ as the subgraph of $H$ by picking each edge in $H$ with probability $\frac{2 c n \log n}{\Delta}$ in a set $E'_s$ (this is slightly different from $H_s$ as in $H_s$ we sample a fixed number of edges). By the above argument, since the minimum cut in $H$ is at least $0.5\Delta$, by Lemma~\ref{lem:cut-connect}, the vertices in $C_i$ are still connected in $H'_s$ as well. 
    On the other hand, recall that by the construction of $H_s$, $E_s$ is a collection of $\frac{10 c n \log n}{\Delta}$ randomly sampled edges without replacement. We call $H_s$ (resp. $H'_s$) is good if $(i)$ $\card{E_s}$ (resp. $\card{E'_s}$) is at most $\frac{10 c n\log n}{\Delta}$, and $(ii)$ any connected component $C_i$ in $H$ which has a $(\delta/4)$-dense vertex is also connected in $H_s$ (resp. $H'_s$). By a straightforward coupling argument, 
    the probability of $H_s$ is good is at least the probability of $H'_s$ is good. By above argument, any connected component $C_i$ in $H$ which has a $(\delta/4)$-dense vertex is also connected in $H'_s$. By Chernoff bound, the number of edges in $E'_s$ is at most $\frac{10 c n \log n}{\Delta}$ with high probability. So $H'_s$ is good with high probability, which means $H_s$ is also good with high probability. 
\end{proof}

Let $C_1$, $C_2$, $\dots$, $C_k$ be the connected components of $H_s$ which have at least $(1-\delta)\Delta$ vertices, and $\Vsparse_{\star}$ be $V \setminus C_1 \cup C_2 \cup \dots \cup C_k$. The following lemma implies Lemma \ref{lem:find-decomposition}

\begin{lemma} \label{lem:stream-decomposition}
   $\Vsparse_{\star},C_1,C_2,\dots,C_k$ is  with high probability a decomposition of $V$ such that:
    \begin{enumerate}[leftmargin=15pt]
        \item\label{app-HSS-p1} $\Vsparse_\star \subseteq \Vsparse_{\delta/4}$, i.e., any vertex in $\Vsparse_\star$ is $(\delta/4)$-sparse. 
        \item\label{app-HSS-p2} For any $i \in [k]$, $C_i$ has the following properties: 
            \begin{enumerate}
                \item\label{app-HSS-p2a}\label{p1} $(1-\delta)\Delta \leq \card{C_i} \leq (1+6\delta)\Delta$; 
                \item\label{app-HSS-p2b} any $v \in C_i$ has at most $7\delta\Delta$ neighbors outside of $C_i$; and
                \item\label{app-HSS-p2c} any $v \in C_i$ has at most $6\delta\Delta$ non-neighbors inside of $C_i$. 
            \end{enumerate}
    \end{enumerate}
\end{lemma}

\begin{proof}
    Since each connected component of $H_s$ is a subset of a connected component of $H$, property \ref{app-HSS-p2} follows from Lemma \ref{lem:app-HSS-decomposition} and the fact that $\card{C_i}\ge (1-\delta)\Delta$. On the other hand, by Lemma \ref{lem:appext-HSS-decomposition} and Lemma \ref{lem:Hs-connect}, any connected component of $H_s$ which contains a $(\delta/4)$-dense vertex has at least $(1-\delta)\Delta$ vertices with high probability. This in turn implies that no $(\delta/4)$-dense vertex is in $\Vsparse_\star$ with high probability. 
\end{proof}

With the decomposition, we are now ready to conclude the proof of Theorem~\ref{thm:streaming}. 

\begin{proof} [Proof of Theorem \ref{thm:streaming}]
   By spending $\Ot(n/\eps^2)$ space throughout the stream, we obtain the decomposition given by Lemma \ref{lem:stream-decomposition}. In the following, we condition on this event and use this decomposition to color $G$ using only the sampled colors for vertices. The proof consists of  three phases, which correspond to the three phases in the proof of Theorem \ref{thm:color-sampling}.
    \paragraph{Coloring Sparse Vertices.} We use the process given by the proof of Lemma \ref{lem:sparse-color} to color the sparse vertices in $\Vsparse_{\star}$ given by Lemma \ref{lem:stream-decomposition}. We run $\GC$ algorithm with $O(\log n)$ rounds. In each round, each vertex $v$ in $\Vsparse_{\star}$ which have not been colored picks a color in $L_1(v)$. Then we check if the color is different from all the vertices in $N(v)$. If so, we color $v$ by the chosen color. So we only need to iterate over the edges in $\Gc$, which takes $O(n\log^2 n)$ time. Hence, this phase of the algorithm takes $O(n\log^3 n)$ time.
    \paragraph{Initial Coloring of Almost Cliques.} In this phase, for each almost-clique $C_i$ given by Lemma \ref{lem:stream-decomposition}, we find a colorful matching (as in Definition~\ref{def:colorful}) given by the the second batch of sampled colors $L_2$. We use the process given by the proof of Lemma \ref{lem:colorful}. For each color, if we can find a pair of vertices $u,v$ such that $(u,v)$ is not in $\Gc$, $u$ and $v$ are not in the colorful matching yet, and $L(u)$ and $L(v)$ both contain the same color, 
    then we add $(u,v)$ with this color to the colorful matching. Hence, this phase also takes $\Ot(n)$ time. 
    
    \paragraph{Final Coloring of Almost-Cliques.} In this phase, we color the remaining vertices inside almost-cliques. To color these vertices, we color the almost-cliques one by one. When coloring almost-clique $C_i$, we construct the {palette-graph} (as in 
    Definition~\ref{def:palette-graph}) between the vertices and the colors, and find a maximum matching of this graph. By Lemma \ref{lem:palette-graph-matching} there is a matching that pairs each vertex to an available color. To construct the \emph{palette-graph}, we need to connect each vertex $v$ with all colors in $L_3(v)$. Then we iterate over the edges of $\Gc$, delete the edges between a vertex and an unavailable color. The construction of the palette-graph for one almost-clique
     takes $O(\Delta \log^2 n)$ time. Finding the matching also require $O(\Delta^{3/2})$ time by using the standard Hopcraft-Karp algorithm~\cite{HopcroftK71} for bipartite matching. 
     There are at most $O(n/\Delta)$ near-cliques in $G$, so this phase takes at most $\Ot(n \sqrt{\Delta})$ time with high probability. 
\end{proof}

\subsection{A Sublinear Time Algorithm for $(\Delta+1)$ Coloring}\label{sec:sub-time}

We now show another application of our palette-sparsification theorem to design sublinear algorithms. Consider the following standard query model for sublinear time algorithms on general graphs (see, e.g., Chapter 10 of Goldreich's 
book~\cite{Goldreich17}): The vertex set of the graph is $V := [n]$ and the algorithm can make the following queries: $(i)$ Degree queries: given $v \in V$, outputs degree $d(v)$ of $v$, $(ii)$ Neighbor queries: given $v \in V$ and $i \leq d(v)$, outputs the $i$-th neighbor
of $v$ (the ordering of neighbors are arbitrary), and $(iii)$ Pair queries: given $u,v \in V$, outputs whether the edge $(u,v)$ is in $E$ or not. We give a sublinear time algorithm (in size of the graph) 
for finding a $(\Delta+1)$ coloring in this query model. 

\begin{theorem}\label{thm:sub-time}
    There exists an algorithm that given a query access to a graph $G(V,E)$ with maximum degree $\Delta$, can find a $(\Delta+1)$ coloring of $G$ with high probability in $\Ot(n\sqrt{n})$ time and queries. 
\end{theorem}

We prove Theorem~\ref{thm:sub-time} by combining two separate algorithms and picking the best of the two depending on the value of $\Delta$. One is the straightforward (deterministic) greedy algorithm that takes $O(n\Delta)$ time to find
a $(\Delta+1)$ coloring. This algorithm only uses neighbor queries. We use this algorithm when $\Delta \leq \sqrt{n}$. The other one is an implementation of $\CMA$ in the query model which takes
$\Ot(n^2/\Delta)$ time. This algorithm only uses pair queries. By using this algorithm when $\Delta \geq \sqrt{n}$, we achieve an $\Ot(n\sqrt{n})$ time algorithm for any graph (with potentially $\Omega(n^2)$ edges), proving Theorem~\ref{thm:sub-time}. 

To implement $\CMA$, we need to specify $(i)$ how to construct the conflict-graph, and $(ii)$ how to find a list-coloring in this conflict graph using the lists $L(\cdot)$. Throughout the proof, we condition on the high probability event 
in Lemma~\ref{lem:meta-algorithm}. The first part of the argument is quite easy as is shown below. 

\begin{claim}\label{clm:sub-time-construct}
    $\Gc(V,\Ec)$ can be constructed in $O(n^2 \cdot \log^2{n}/\Delta)$ queries and time.
\end{claim}
\begin{proof}
    As the vertices are known, we only need to construct the edges $\Ec$. In order to do this, we simply query all pairs between vertices inside each set $\CC_c$ for $c \in [\Delta+1]$. 
    By Lemma~\ref{lem:meta-algorithm}, size of each $\CC_c$ is $O(n\log{n}/\Delta)$ and so we need $O(n^2\log^{2}{n}/\Delta^2 \cdot (\Delta+1)) = O(n^2\log^2{n}/\Delta)$ queries. 
\end{proof}

Claim~\ref{clm:sub-time-construct} is already sufficient to obtain an $\Ot(n^2/\Delta)$ \emph{query} (but not time) algorithm: by Lemma~\ref{lem:meta-algorithm}, $\CMA$ outputs the correct answer by finding a list-coloring of $\Gc$
and accessing $\Gc$ does not require further queries to $G$. However, finding such a list-coloring problem in general is NP-hard and hence to find this coloring in sublinear time, we need
to design an algorithm which further queries the graph $G$ to obtain additional information for performing the coloring. The idea is as in the previous section by finding an (approximate) extended-HSS decomposition of $G$
using a small number of queries and then constructing the list-coloring of $\Gc$ using this additional information. 

\paragraph{List-Coloring $\Gc$ in Sublinear Time in the Query Model.}
Just as in the case of the streaming model, we first give an algorithm which finds an approximate extended-HSS decomposition of $G$ in $\tilde{O}(n^2/\Delta)$ time. We then use an arguments similar to Theorem \ref{thm:color-sampling} to color graph $G$ also 
in $\tilde{O}(n^2/\Delta)$ time. The parameter $\delta$ below is as defined in Section~\ref{sec:streaming}. 

To find an approximate extended-HSS decomposition, we use similar ideas as in Section~\ref{sec:streaming}. We first give a data structure to 
determine whether a pair of vertices are potential friends (see Section~\ref{sec:streaming} for the definition). 
We then use this data structure to find dense vertices. The arguments are similar to Lemma \ref{lem:detect-potential-friend} and Lemma \ref{lem:detect-eps-dense}, but we now implement them in sublinear time using only pair queries. 
Once we have these two lemmas, we define the graph $H$ on dense vertices as in Section~\ref{sec:streaming} and find a random subgraph $H_s$ of $H$ by random sampling edges, and prove that the subgraph has the same connected components
as $H$, which allows us to recover the decomposition.

The following two lemmas are analogies to Lemma \ref{lem:detect-potential-friend} and Lemma \ref{lem:detect-eps-dense}, but we achieve them in sublinear time in query model.

\begin{lemma} \label{lem:time-potential-friend}
    There exists a data structure that can be created in $\tilde{O}(n^2/\Delta)$ time using pair queries 
    for answering the following queries: Given a pair of vertices $(u,v) \in V \times V$,
    \begin{itemize}
        \item outputs \Yes if $u$ and $v$ are $\delta$-potential friend;
        \item outputs \No if $u$ and $v$ are not $2\delta$-potential friend; and
        \item the answer can be arbitrary otherwise. 
    \end{itemize}
    For each query pair $u, v$, the data structure can be used to recover the correct answer with high probability in $O(\log n)$ time. 
\end{lemma}

\begin{proof}
    We use the same notation as in Lemma~\ref{lem:detect-potential-friend}. 
    To construct the data structure, pick a set of $S$ of vertices by choosing each vertex with probability $p=\frac{10\log n}{\delta^2\Delta}$. For each vertex $v\in S$, use the pair queries to find all its neighbors in $G$, and construct a linked list for each vertex $u\in V$ to store all the neighbors of $u$ in $S$. Now to answer the query of a pair of vertices $u$ and $v$, use the list of neighbors in the data structure to calculate $c_{u,v}:= \card{N(u) \cap N(v) \cap S}$, i.e., the number of common neighbors of $u$ and $v$ in $S$. Output \Yes if and only if $c_{u,v}\ge (1-1.5\delta)\Delta \cdot p$. 

    Using the proof of Lemma \ref{lem:detect-potential-friend}, the answer of each query is correct with high probability. So we only need to show the running time of the algorithm. By Chernoff bound, the number of vertices chosen in $S$ is at most $O(n \log n/\Delta)$ with high probability. We query all pairs of vertices which contains the vertices in $S$. It takes $O(n^2 \log n/\Delta)$ queries. We can construct the linked lists during the queries. So we only need $O(n^2 \log n/\Delta)$ time to construct the data structure with high probability. For each vertex, there are at most $\Delta$ neighbors in $G$ and each neighbor is chosen in $S$ with probability $p$. Hence, by Chernoff bound, there are at most $O(p\Delta) = O(\log n)$ neighbors in $S$ with high probability. So for each query, we need at most $O(\log n)$ time to calculate $c_{u,v}$ and answer the query with high probability. 
\end{proof}

From now on, we condition on the event in Lemma~\ref{lem:time-potential-friend} for all pair of vertices, which still happens with high probability by a union bound. 
We can now use the above data structure to find $\delta$-dense vertices approximately.
\begin{lemma}\label{lem:time-eps-dense}
    There exists an algorithm that uses $\Ot(n^2/\Delta)$ time and pair queries and with high probability outputs a set $D$ of vertices such that: 
    \begin{itemize}
        \item any vertex in $D$ is at least $2\delta$-dense;
        \item all $(\delta/2)$-dense vertices belong to $D$. 
    \end{itemize}
\end{lemma}

\begin{proof}
	The proof is again similar to the proof of Lemma~\ref{lem:detect-eps-dense}. We add each edge in the graph to a set $E_s$ with probability $p=\frac{10\log n}{\delta^2 \Delta}$ independently. To construct the set $E_s$, we simply need
	to sample every edge-slot in the graph with probability $p$ and then query the sampled edge-slot. With high probability, this can be done in $O(n^2 \cdot p) = O(n^2\log n/ \Delta)$ pair queries (to sample this set, we can sample random variables by
	geometric distribution with parameter $p$ to determine the next pair to query). 

    For each edge $(u,v) \in E_s$, we use the data structure in Lemma~\ref{lem:time-potential-friend} to determine if they are potential friends. We discard $(u,v)$ from $E_s$ if the answer of the query is \No. We can do it in $O(n^2 \log^2 n /\Delta)$ time with high probability since each query takes $O(\log n)$ time (as we conditioned on the event in Lemma~\ref{lem:time-potential-friend}). Add a vertex $v$ to $S$ if there are at least $(1-1.5 \delta)\Delta p$ edges in $E_s$ which is incident on $v$. By the same argument as in the proof of Lemma \ref{lem:detect-eps-dense},  the output $S$ satisfies the requirements in the lemma with high probability.
\end{proof}

For the rest of the proof, we further condition on the event in Lemma~\ref{lem:time-eps-dense}. Recall the definition of the graph $H$ from Section~\ref{sec:streaming}. 
We can now construct a random subgraph $H_s$ of $H$. Let $E_s$ be the edge set by choosing each edge in the graph with probability $p = \frac{4c\log n}{\Delta}$, where $c$ is the constant $c(b)$ in Lemma \ref{lem:Hs-connect} with $b=3$. 
Let $H_s$ be the graph that only contains vertices in $D$ (the output set of Lemma \ref{lem:time-eps-dense}), with the edge set: 
\begin{align*}
E(H_s) = & \{(u,v)\in E_s \mid \text{data structure in Lemma \ref{lem:time-potential-friend}} \text{returns \Yes to query $(u,v)$}\}.
\end{align*}
To construct the graph $H_s$, we need $O(n^2p)=O(n^2\log n/ \Delta)$ time with high probability to construct $E_s$, and $O(n^2\log^2 n/\Delta)$ time to query each edge in $E_s$ using the data structure in Lemma \ref{lem:time-potential-friend}.

Let $C_1$, $C_2$, $\dots$, $C_k$ be the connected components of $H_s$ which have at least $(1-\delta)\Delta$ vertices, and $\Vsparse_{\star}$ be $V \setminus C_1 \cup C_2 \cup \dots \cup C_k$. The following lemma directly follows from the proof of Lemma \ref{lem:stream-decomposition} with Lemma \ref{lem:time-potential-friend} and Lemma \ref{lem:time-eps-dense}.

\begin{lemma} \label{lem:time-decomposition}
    With high probability, we can compute a decomposition of $V=\Vsparse_{\star}\cup C_1 \cup C_2 \cup \dots \cup C_k$ in $\Ot(n^2/\Delta)$ time and pair queries such that:
    \begin{enumerate}[leftmargin=15pt]
        \item\label{app-HSS-p1} $\Vsparse_\star \subseteq \Vsparse_{\delta/4}$, i.e., any vertex in $\Vsparse_\star$ is $(\delta/4)$-sparse. 
        \item\label{app-HSS-p2} For any $i \in [k]$, $C_i$ has the following properties: 
            \begin{enumerate}
                \item\label{time-HSS-p2a}\label{p1} $(1-\delta)\Delta \leq \card{C_i} \leq (1+6\delta)\Delta$; 
                \item\label{time-HSS-p2b} any $v \in C_i$ has at most $7\delta\Delta$ neighbors outside of $C_i$; and
                \item\label{time-HSS-p2c} any $v \in C_i$ has at most $6\delta\Delta$ non-neighbors inside of $C_i$. 
            \end{enumerate}
    \end{enumerate}
\end{lemma}

Proof of Theorem~\ref{thm:sub-time} can now be concluded exactly as in the proof of Theorem~\ref{thm:streaming}, as we show in that proof that starting from the decomposition, one can construct the coloring in $\Ot(n\sqrt{n})$ time.

\subsection{An MPC Algorithm for $(\Delta+1)$ Coloring}\label{sec:mpc}

This section contains yet another application of our palette-sparsification theorem to design sublinear algorithms, namely a massively parallel (MPC) algorithm for $(\Delta+1)$ coloring. 
In the MPC model of~\cite{BeameKS13} (see also~\cite{KarloffSV10,GoodrichSZ11,AndoniNOY14,BeameKS13}), 
the input is partitioned across multiple machines which are inter-connected via a communication network. The computation proceeds in synchronous rounds. 
During a round each machine runs a local algorithm on the data assigned to the machine. No communication between machines is allowed during a round. Between rounds, machines are allowed to communicate so long as each machine sends or receives a 
communication no more than its memory. Any data output from a machine must be computed locally from the data residing on the machine and initially the input data is distributed across machines adversarially. 
The goal is to minimize the total number of rounds subject to a small (sublinear) memory per machine. 

We show that $\CMA$ can be easily implemented in this model also and prove the following theorem.  
 
\begin{theorem}\label{thm:mpc}
	There exists a randomized MPC algorithm that given a graph $G(V,E)$ with maximum degree $\Delta$ can find a $(\Delta+1)$ coloring of $G$ with high probability in $O(1)$ MPC rounds on machines of 
	memory $\Ot(n)$. Furthermore, if the machines have access to \emph{public randomness}, the algorithm only requires \emph{one} MPC round. 
\end{theorem}

In the following, we assume that the machines have access to public randomness and show how to solve the problem in only one MPC round (this setting is similar to the result of~\cite{AhnGM12Linear} for the connectivity problem and is of independent 
interest). We then show that by spending $O(1)$ additional rounds, we can remove the assumption of public randomness. 

The proof of this theorem is very similar to that of Theorem~\ref{thm:streaming} and uses the close connection between dynamic streaming algorithms (in particular linear sketching algorithms) and MPC algorithms (see, e.g.~\cite{AhnGM12Linear,AhnG15}). 
As before, if we do not insist on achieving a polynomial time algorithm, proving Theorem~\ref{thm:mpc} from Lemma~\ref{lem:meta-algorithm} is straightforward: we sample the color classes $\CC_1,\ldots,\CC_{\Delta+1}$ using public randomness, 
and every machine sends its edges in $\Gc$ to a central designated machine, called the coordinator. 
As the total number of edges in $\Gc$ is $\Ot(n)$ by Lemma~\ref{lem:meta-algorithm}, we only need $\Ot(n)$ memory on the coordinator. The coordinator machine can then locally find a list-coloring of the graph $\Gc$ and 
find the $(\Delta+1)$ coloring by Lemma~\ref{lem:meta-algorithm}. 

We now briefly show how to make this algorithm polynomial time. The proof of this part is identical to that of Theorem~\ref{thm:streaming}: we can construct a decomposition of the graph $G$ in Lemma~\ref{lem:stream-decomposition}
exactly as in Section~\ref{sec:streaming} on the coordinator machine and then run the polynomial time algorithm in the proof of Theorem~\ref{thm:streaming} on this decomposition to solve the problem. In order to see this, note that
the decomposition in Lemma~\ref{lem:stream-decomposition} relied on vertex sampling (to implement the data structure in Lemma~\ref{lem:detect-potential-friend}) and edge sampling (to implement the data structure in Lemma~\ref{lem:detect-eps-dense}
as well as defining and finding the connected components of the graph $H_s$), both of which can be done trivially in the MPC model using public randomness. 

Finally, we show how to remove the public randomness. We first dedicate one machine $M_v$ to each vertex $v$ of the graph and spend the first round to send all the edges incident on $v$ to the machine $M_v$. This can be done on machines of memory 
$O(\Delta)$. In the next round, each machine $v$ samples the set of colors $L(v)$ for $v$ and sends this information to all the machines $M_u$ where $(u,v)$ is an edge in the graph. This can again be done with $O(\Delta \cdot \polylog{(n)})$ size messages and 
hence on machines of memory $\Ot(n)$. The machines can now send all edges in $\Ec$ to a central coordinator
and the coordinator can construct the graph $\Gc$. Finally, to find the decomposition in Lemma~\ref{lem:stream-decomposition} also, each machine $M_v$ can do the sampling of vertex $v$ locally and send this vertex and all its incident edges (as needed
in the proof of Lemma~\ref{lem:detect-potential-friend}) to the central coordinator; the edge sampling can also be done trivially as each machine holding an edge can decide to sample it with probability $p$ and send it to the coordinator independently (as needed in 
the proof of Lemma~\ref{lem:detect-eps-dense} and also in finding connected components of $H_s$). All in all, we only need $O(1)$ MPC rounds to run the algorithm.

\section{Lower Bounds for Graph Coloring and Related Problems}\label{sec:lower-bounds}

In this section, we prove new lower bounds for maximal independent set (MIS) and maximal matching problems in the considered models to contrast the complexity of the $(\Delta+1)$ coloring problem with these two closely related problems

\begin{theorem}[\textbf{Streaming MIS}] \label{thm:lower-MIS}
	Any single-pass streaming algorithm that is able to output any arbitrary MIS of the input graph in insertion-only streams with sufficiently large constant probability  requires $\Omega(n^2)$ space. 
\end{theorem}

Theorem~\ref{thm:lower-MIS} implies that unlike the $(\Delta+1)$ coloring problem, one cannot obtain any non-trivial single-pass streaming algorithm for the maximal independent set even in insertion-only streams. 
We remark that a result very similar to Theorem~\ref{thm:lower-MIS} in our paper was obtained independently and concurrently by Cormode~\etal~\cite{CormodeDK18} (see Theorem~3). 

\begin{theorem}[\textbf{Sublinear-Time Maximal Matching}] \label{thm:lower-matching}
	Any algorithm (possibly randomized) that can output a maximal matching of an input graph with sufficiently large constant probability requires $\Omega(n^2)$ queries to the graph. 
\end{theorem}

Theorem~\ref{thm:lower-matching} implies that unlike the $(\Delta+1)$ coloring problem, one cannot hope to obtain any non-trivial sublinear-time algorithm for the maximal matching problem. 

\medskip

We also prove a lower bound for $(\Delta+1)$ coloring problem in the query model that shows that our algorithm in Theorem~\ref{thm:sub-time} achieves optimal query/time complexity (up to logarithmic factors). This lower bound
holds even for algorithms that are required to output an $O(\Delta)$ coloring not $(\Delta+1)$. 

\begin{theorem}[\textbf{Sublinear Time Vertex Coloring}] \label{thm:lower-color}
	For any fixed constant $c > 1$, any algorithm (possibly randomized) that can output a $(c \cdot \Delta)$ coloring of an input graph with sufficiently large constant probability requires $\Omega(n\sqrt{n})$ queries to the graph. 
\end{theorem}

We prove Theorems~\ref{thm:lower-MIS}, \ref{thm:lower-matching}, and~\ref{thm:lower-color}, in Sections~\ref{sec:lower-MIS},~\ref{sec:lower-matching}, and~\ref{sec:lower-color}, respectively.

\newcommand{\MIS}{\ensuremath{\textnormal{\textsf{MIS}}}\xspace}

\newcommand{\Index}{\ensuremath{\textnormal{\textsf{Index}}}\xspace}

\newcommand{\kstar}{\ensuremath{k^{*}}}

\newcommand{\istar}{\ensuremath{i^{*}}}

\renewcommand{\jstar}{\ensuremath{j^{*}}}

\newcommand{\MM}{\ensuremath{\mathcal{M}}}

\subsection{A Lower Bound for Single-Pass Streaming MIS}\label{sec:lower-MIS}

We prove Theorem~\ref{thm:lower-MIS} in this section. Throughout this section, we define $\MIS$ as the two-player communication problem in which the edge-set $E$ of a graph $G(V,E)$ is partitioned between Alice and Bob, and their goal is to find \emph{any} MIS of the graph $G$. To prove
Theorem~\ref{thm:lower-MIS}, it suffices to lower bound the \emph{one-way} communication complexity of \MIS. We do this using a reduction from the \Index problem. In \Index, Alice is given a bit-string $x \in \set{0,1}^{N}$, Bob is given an index $\kstar \in [N]$, 
and the goal is for Alice to send a message to Bob so that Bob outputs $x_{\kstar}$. It is well-known that one-way communication complexity of \Index is $\Omega(N)$ bits~\cite{Ablayev93,KremerNR95}. 

To continue, we need a simple definition. Let $n := \sqrt{N}$ and throughout the proof, fix any arbitrary bijection $\sigma : [N]\rightarrow [n] \times [n]$. For any bit-string $x \in \set{0,1}^{N}$, and sets $A := \set{a_1,\ldots,a_n}$ and $B:= \set{b_1,\ldots,b_n}$, we 
define the \emph{$x$-graph} $G$ on vertices $A \cup B$ as a bipartite graph consisting of all edges $(a_i,b_j)$ where $x_k = 1$ and $\sigma(k) = (i,j)$. 

 \medskip
 
 \paragraph{Reduction.} Given an instance $(x,\kstar)$ of \Index:

\begin{enumerate} 
\item Alice and Bob construct the following graph $G(V,E)$ with \emph{no} communication: 

\begin{itemize}
	\item The vertex set $V$ of the graph $G$ is the union of the following four sets of vertices: 
		\begin{align*}
		&A_1 := \set{a^1_1,\ldots,a^1_n}, B_1:= \set{b^1_1,\ldots,b^1_n}, \\
		&A_2:=\set{a^2_1,\ldots,a^2_n}, B_2:=\set{b^2_1,\ldots,b^2_n}.
		\end{align*}
	\item  
	
	Alice forms two $x$-graphs $G_1$ and $G_2$ on $A_1 \cup B_1$ and $A_2 \cup B_2$, respectively, using her bit-string $x$ and adds these edges to $E$. 
	
	\smallskip
	
	\item Let $(\istar,\jstar) = \sigma(\kstar)$ where $\kstar$ is the input index of Bob. Bob adds the following three sets of edges to $E$: 
	\begin{align*}
		&\set{(a^{1}_i,a^{1}_{j}), (a^{2}_i,a^{2}_{j}), (a^{1}_i,a^{2}_j)} \tag{for all $i \neq \istar \in [n]$ and $j \neq \istar \in [n]$}, \\
		&\set{(b^{1}_i,b^{1}_{j}), (b^{2}_i,b^{2}_{j}), (b^{1}_i,b^{2}_j)} \tag{for all $i \neq \jstar \in [n]$ and $j \neq \jstar \in [n]$}, \\
		&\set{(a^{1}_i,b^{2}_{j}), (a^{2}_i,b^{1}_{j})} \tag{for all $i \neq \jstar \in [n]$ and $j \neq \jstar \in [n]$}.
	\end{align*}

\end{itemize}
	
\item Alice and Bob then compute any MIS $\MM$ of the graph $G$ using the best protocol for \MIS. Bob outputs $x_{\kstar} = 0$ if either both $a^{1}_{\istar},b^{1}_{\jstar} \in \MM$ or both $a^{2}_{\istar},b^{2}_{\jstar} \in \MM$.

\end{enumerate}

This finalizes the description of the reduction. It is immediate to verify that the communication cost of this protocol is at most as large as the communication complexity of \MIS. In the following two claims, we prove the correctness of this protocol. 

\begin{claim}\label{clm:one-way-trivial}
	Let $G$ be the graph constructed by Alice and Bob, and $\MM$ be any MIS of $G$. Then, if $x_{\kstar}=1$, $\MM$ cannot contain both $a^{1}_{\istar},b^{1}_{\jstar}$, neither both $a^{2}_{\istar},b^{2}_{\jstar}$.
\end{claim}
\begin{proof}
	If $x_{\kstar}=1$, then $(a^{1}_{\istar},b^{1}_{\jstar}) \in G_1$ and $(a^{2}_{\istar},b^{2}_{\jstar}) \in G_2$ by construction of the $x$-graph. 
\end{proof}

\begin{claim}\label{clm:one-way-non-trivial}
	Let $G$ be the graph constructed by Alice and Bob. Assuming $x_{\kstar}=0$, \emph{any} MIS $\MM$ of $G$ contains either both $a^{1}_{\istar},b^{1}_{\jstar}$ or both $a^{2}_{\istar},b^{2}_{\jstar}$. 
\end{claim}
\begin{proof}
	Suppose towards a contradiction that $\MM$ does not satisfy the assertion of the claim statement. 
	It is easy to see that no collection of two vertices from $a^{1}_{\istar},b^{1}_{\jstar},a^{2}_{\istar},b^{2}_{\jstar}$ can be a dominating set in $G$ and an MIS (any MIS of a graph is a dominating set). This implies that $\MM$
	necessarily contains a vertex $v \in G \setminus \set{a^1_{\istar},b^1_{\jstar},a^2_{\istar},b^2_{\jstar}}$. For simplicity, let us assume that $v \in A_1$; the other cases hold 
	by symmetry. We know that Bob has connected $v$ to all vertices in $A_1 \cup A_2 \setminus \set{a^1_{\istar},a^2_{\istar}}$ as well as vertices in $B_2 \setminus b^2_{\jstar}$. As such, none of these vertices can be part of $\MM$. 
	
	Now consider vertices $a^2_{\istar},b^2_{\jstar}$.  All ``potential'' neighbors of $a^2_{\istar}$ in $G$ are vertices in $B_2$ and hence except for $b^2_{\jstar}$ none of them can be in $\MM$. Similarly, all potential neighbors of 
	$b^{2}_{\jstar}$ in $G$ are vertices in $A_2$ and again except for $a^2_{\istar}$ none of them can be in $\MM$. Maximality of $\MIS$ plus the fact that there are no edges between $a^2_{\istar}$ and $b^{2}_{\jstar}$ (as $x_{\kstar}=0$), 
	implies that both $a^2_{\istar}$ and $b^{2}_{\jstar}$ are in $\MM$, a contradiction. 
\end{proof}
\noindent
We are now ready to prove Theorem~\ref{thm:lower-MIS}.
\begin{proof}[Proof of Theorem~\ref{thm:lower-MIS}]
	Let $\prot$ be any $1/3$-error one-way protocol for \MIS. By Claims~\ref{clm:one-way-trivial} and~\ref{clm:one-way-non-trivial}, we obtain a protocol for \Index which errs with probability at most $1/3$ and has communication
	cost at most equal to cost of $\prot$. By the $\Omega(N)$ lower bound on the one-way communication complexity of \Index, we obtain that communication cost of $\prot$ must be $\Omega(N) = \Omega(n^2)$. As the total number
	of vertices in the graph is $O(n)$, we obtain an $\Omega(n^2)$ lower bound on the communication complexity of \MIS. Theorem~\ref{thm:lower-MIS} now follows from this as one-way communication complexity lower bounds the space complexity
	of single pass streaming algorithms. 
\end{proof}

\subsection{A Lower Bound for Sublinear Time Maximal Matching}\label{sec:lower-matching}
We prove Theorem~\ref{thm:lower-matching} in this section. Recall the definition of the query model from Section~\ref{sec:sub-time}. 
We will show that there exists a family of $n$-vertex graphs such that any randomized algorithm for finding a maximal matching requires $\Omega(n^2)$ queries in expectation and use this to prove Theorem~\ref{thm:lower-matching}. 
We will show that this lower bound holds even for bipartite graphs.
 
We will denote our input graph by $G(V,E)$ where $V$ is bipartitioned into sets $L$ and $R$. We will assume that the algorithm is given degrees of all the vertices for free.
As such, we simply consider the following two types of queries:

\begin{itemize}
\item ${\mathbf Q_1(u,i)}$ (neighbor query) : given any vertex $u \in V$ and an integer $i \in [1 .. \deg(u)]$, this query returns the $i$-th neighbor of $u$ in the adjacency list of $u$ (the ordering of the list is arbitrary).

\item 
${\mathbf Q_2(u,v)}$ (pair query): given a pair of vertices $u,v \in V$, this query returns $0/1$ to indicate absence or presence of the edge $(u,v)$.

\end{itemize}

By Yao's minimax principle~\cite{Yao87}, it suffices to create a distribution over $n$-vertex graphs such that any deterministic algorithm requires $\Omega(n^2)$ queries to find a maximal matching on this distribution.
A graph from our distribution is generated as follows.
The set $L$ is partitioned into sets $L_1$ and $L_2$, and the set $R$ is partitioned into sets $R_1$ and $R_2$ with $|L_1| = |R_1| = n/6$, and $|L_2| = |R_2| = 5n/6$. Each vertex $u$ in $L$ is connected to every vertex in $R_1$, and each vertex $v$ in $R$ is connected to each vertex in $L_1$.
In addition, we add to $G$ a random perfect matching $M$ between $L_2$ and $R_2$. Finally, the adjacency list of each vertex is created by taking a random permutation of its neighbors. 
This completes the description of the process that generates our input graph $G(V,E)$.

Note that the final graph $G$ contains a perfect matching of size $n$. Thus any maximal matching in $G$ includes at least $n/3$ edges from $M$, and any maximal matching algorithm needs to discover at least $n/3$ edges in $M$. The heart of our lower bound argument is to show that this requires $\Omega(n^2)$ queries in expectation.

We assume that the algorithm is explicitly given the partition of vertices into sets $L_1, L_2, R_1,$ and $R_2$. This means that the algorithm already knows adjacency list of each vertex $u \in L_1$ and each vertex $v \in R_1$, as well as the answer to each query $Q_2(u,v)$ whenever $u \in L$ and $v \in R_1$ or whenever $u \in L_1$ and $v \in R$. 
Thus the only potentially useful queries for any algorithm are of the form $Q_1(u,i)$ for some $u \in L_2 \cup R_2$ and $i \in [n/6 +1]$, and $Q_2(u,v)$ for $u \in L_2$ and $v \in R_2$.

 We say that an edge $(u,v) \in M$ has been {\em discovered} if the set of queries performed thus far uniquely identify the edge $(u, v)$ to be in $M$. 
For any discovered edge $(u,v) \in M$, we will say the vertices $u$ and $v$ have been {\em discovered};
any vertex in $L_2 \cup R_2$ that has not been discovered is called an {\em undiscovered} vertex. 
We are now ready to prove the main result below.

\begin{lemma}
\label{lem:matching_queries}
Given a graph $G$ generated by the distribution described above, any deterministic algorithm requires $\Omega(n^2)$ queries in expectation to discover $n/3$ edges in $M$.
\end{lemma}
\begin{proof}
After $t$ queries have been made by the algorithm, let $L_U(t) \subseteq L_2$ and $R_U(t) \subseteq R_2$ denote the set of undiscovered vertices in $L_2$ and $R_2$ respectively. Let $E(t) \subseteq L_U(t) \times R_U(t)$ denote the set of edge slots that have not yet been queried/discovered. Note that by our process for generating $G$, the undiscovered edges in $M$ correspond to a random perfect matching between $L_U(t)$ and $R_U(t)$ that is entirely supported on $E(t)$. 

We will analyze the performance of any algorithm by partitioning the queries into phases. The first query by the algorithm starts the first phase, and a phase ends as soon as an edge in $M$ has been discovered. 
Let $Z_i$ be a random variable that denotes the number of queries performed in phase $i$ of the algorithm. Thus we wish to analyze the $E[ \sum_{i=1}^{n/3} Z_i]$. The following lemma is crucial to our analysis. 

\begin{lemma}
\label{lem:random_matching}
Let $G'(L' \cup R',E')$ be an arbitrary bipartite graph such that $|L'| = |R'| = N$, and each vertex in $G'$ has degree at least $2N/3$. Then for any edge $e = (u,v) \in E'$, the probability that $e$ is contained in a uniformly at random chosen perfect matching in $G'$ is at most $3/N$. 
\end{lemma}
\begin{proof}
It is easy to verify that such a $G'$ always contains a perfect matching. Let us fix any perfect matching $M$ that contains $e = (u,v)$. We will show that we can transform $M$ into perfect matchings $M_1, M_2, ..., M_p$ for $p \ge N/3$ such that each $M_i$ identical to $M$ except that the edge $(u,v)$ and some edge $(a_i,b_i)$ in $M$ are replaced with edges $(u, b_i)$ and $(v, a_i)$ in $M_i$. Furthermore, for any other perfect matching $M'$ containing the edge $e = (u,v)$, the set of perfect matchings $M'_1, M'_2, ..., M'_q$ created by above-mentioned transformation are pairwise distinct from $M_1, M_2, ..., M_p$. It then follows that for each perfect matching containing edge $e$, we can exhibit a unique set of $N/3$ perfect matchings that do not contain edge $e$. Hence the probability that a random perfect matching in $G$ contains the edge $e$ is at most $3/N$ as claimed.

We now describe the transformation. Since each vertex in $G$ has degree at least $2N/3$, it follows that there are at least $p \ge N/3$ edges in $M$, say, $(a_1,b_1), (a_2, b_2), ..., (a_p, b_p)$ such that $u$ is adjacent to each of $b_1, b_2, ..., b_p$ and $v$ is adjacent to each of $a_1, a_2, ..., a_p$. We can then create matching $M_i$ by replacing the edge $(u,v)$ and some edge $(a_i,b_i)$ in $M$ with edges $(u, b_i)$ and $(v, a_i)$ in $M_i$. All other edges remain the same. 

Finally, we show that for any other perfect matching $M'$ containing the edge $e = (u,v)$, the set of perfect matchings $M'_1, M'_2, ..., M'_q$ created by above-mentioned transformation are pairwise distinct from $M_1, M_2, ..., M_p$. To see this, consider any edge $(x,y) \in M \setminus M'$. None of the matchings $M'_1, M'_2, ..., M'_q$ contain the edge $(x,y)$. On the other hand, with at most one possible exception, each of the matchings $M_1, M_2, ..., M_p$ contains the edge $(x,y)$. The only possible exception occurs when the edge $(x,y)$ corresponds to an edge $(a_i, b_i)$ for some matching $M_i$. In that case, $M_i$ contains edges $(u, y)$ and $(x,v)$, neither of which can be present in any of the matchings $M'_1, M'_2, ..., M'_q$.

This completes the proof of the lemma.
\end{proof}

For any vertex $w \in (L_U(t) \cup R_U(t))$, we say that {\em $Q_1$-uncertainty} of $w$ is $d$ if there are at least $d$ entries in the adjacency list of $w$ that have not yet been probed. Similarly, for any vertex $w \in (L_U(t) \cup R_U(t))$, we say that {\em $Q_2$-uncertainty} of $w$ is $d$ if there are at least $d$ edges in $E(t)$ that are incident on $w$.

At time $t$, we say a vertex $w \in (L_U(t) \cup R_U(t))$ is {\em bad} if either $Q_1$-uncertainty of $w$ is less than $n/12$ or $Q_2$-uncertainty of $w$ is less than $3n/4$. 
Note that if at some time $t$ none of the vertices in $(L_U(t) \cup R_U(t))$ are bad then
in the next $n/24$ time steps, any single $Q_1$ query succeeds in discovering an edge in $M$ with probability at most $24/n$. 
Also, for the next $n/24$ time steps, the degree of any vertex in $L_U(t) \cup R_U(t)$ in $E(t)$ remains above $2n/3$. Thus by Lemma~\ref{lem:random_matching}, the probability that any single $Q_2$ query made during the first $n/12$ queries in the phase succeeds (in discovery of a new edge in $M$) is at most $3/n$. 

We say a phase is {\em good} if at the \underline{start} of the phase, there are no bad vertices, and the phase is {\em bad} otherwise.

\begin{proposition}
The expected length of a good phase is at least $n/96$.
\end{proposition}
\begin{proof}
If at the start of the phase $i$, no vertex is bad, then for the next $n/24$ time steps, the probability of success for any $Q_1/Q_2$ query is at most $24/n$. Thus the expected number of successes (discovery of a new edge in $M$) in the first $n/48$ time steps in a phase is at most $1/2$. By Markov's inequality, it then follows that with probability at least $1/2$, there are no successes among the first $n/48$ queries in a phase. Thus the expected length of the phase is  $\geq n/96$. 
\end{proof}

Note that if all phases were good, then it immediately follows that the expected number of queries to discover $n/3$ edges is $\Omega(n^2)$. To complete the proof, it remains to show that most phases are good. For ease of analysis, we will give the algorithm additional information for free and show that it still needs $\Omega(n^2)$ queries in expectation even to discover the first $n/24$ edges in $M$.

Whenever algorithm starts a bad phase, we immediately reveal to the algorithm an undiscovered edge in $M$ that is incident on an arbitrarily chosen bad vertex. Thus each bad phase is guaranteed to consume a bad vertex (i.e., make the bad vertex discovered and hence remove it from further consideration). On the other hand, to create a bad vertex $w$, one of the following two events needs to occur:
(a) either we have done at least $(n/12)$ $Q_1$ queries incident on $w$, or 
(b) the number of discovered edges in $M$ plus the number of $Q_2$ queries is at least $5n/6 - 3n/4 = n/12$.

Since we are restricting ourselves to analyzing the discovery of first $n/24$ edges in $M$, any vertex $w$ that becomes bad due to (b) above requires at least $n/24$ $Q_2$ queries incident on it. Thus to create $K$ bad vertices in the first $n/24$ phases, we 
need to perform at least $(K \cdot (n/24))/2$ queries; here the division by $2$ accounts for the fact that each $Q_2$ query reduces uncertainity for two vertices. It now follows that if the algorithm encounters at least $n/48$ bad phases among the first $n/24$ 
phases, then $K \ge n/48$ and hence it must have already performed $n^2/(48)^2$ queries. Otherwise, at least $n/48$ phases among the first $n/24$ phases are good, implying that the expected number of queries is at least $(n/48) \cdot (n/96)$. This completes 
the proof of Lemma~\ref{lem:matching_queries}. 
\end{proof}

\begin{proof}[Proof of Theorem~\ref{thm:lower-matching}]
Suppose towards a contradiction that $\alg$ is a query algorithm that uses $o(n^2)$ queries and find a maximal matching of the graph with probability $1-o(1)$. 
Consider the distribution of instances introduced in this section and notice that any maximal matching necessarily needs to find $n/3$ edges from $M$ (as any maximal matching is a $2$-approximate matching). 
We run $\alg$ on the distribution of this section and if it did not discover $n/3$ edges from $M$, we simply query all edges of the graph. As such, the expected query complexity of this algorithm before finding $n/3$ edges from 
$M$ is $o(n^2) \cdot (1-o(1))$ (if it finds the maximal matching) plus $n^2 \cdot o(1)$ (if it did not find the maximal matching and we queried all graph) which is $o(n^2)$. By fixing the randomness of this algorithm against this distribution, 
i.e., using (the easy direction of) Yao's minimax principle, we obtain a deterministic algorithm that makes $o(n^2)$ queries in expectation and finds $n/3$ edges from $M$. This contradicts Lemma~\ref{lem:matching_queries}, which finalizes the proof. 
\end{proof}

\subsection{A Lower Bound for $(\Delta+1)$ Vertex Coloring}\label{sec:lower-color}

In this section, we prove Theorem~\ref{thm:lower-color} by showing that there exists a family of $n$-vertex graphs such that for any constant $c > 1$, any randomized algorithm that outputs a valid $(c \cdot \Delta)$ coloring on this family with probability at least $1 - o(1)$, requires $\Omega(n\sqrt{n})$ queries. 
For ease of notation, we will work with graphs with $2n$ vertices and focus on proving that finding a $(1.99 \Delta)$ coloring requires $\Omega(n\sqrt{n})$ queries; essentially the same proof argument also implies an indetical lower bound for a $(c \cdot \Delta)$ coloring.

The maximum degree $\Delta$ of each graph in our family will be $\sqrt{n}+1$. 
Since for large enough $n$, $1.99(\sqrt{n}+1) < 2\sqrt{n}$, it suffices to show that any randomized algorithm for finding a $(2\sqrt{n})$ coloring with large constant probability requires $\Omega(n\sqrt{n})$ queries. We use the same query model defined in Section \ref{sec:lower-matching}.

By Yao's minimax principle~\cite{Yao87}, it suffices to create a distribution over graphs with $2n$ vertices such that any deterministic algorithm requires $\Omega(n\sqrt{n})$ queries to find a valid coloring with probability at least $1 - o(1)$. A graph from our distribution is generated as follows. 
The vertex set is divided into $\sqrt{n}+1$ sets $V_0,V_1,\dots,V_{\sqrt{n}}$ where $V_0$ has $n$ vertices, and each set $V_i$, for $1 \le i \le n$, has exactly $\sqrt{n}$ vertices. Furthermore, the vertices in $V_0$ are partitioned into $\sqrt{n}$ sets $V'_1,V'_2,\dots,V'_{\sqrt{n}}$ where each set has $\sqrt{n}$ vertices. For any $1 \le i \le \sqrt{n}$, each vertex in $V_i$ is connected to every vertex in $V'_i$. Finally, we pick a \emph{random perfect matching} $M$ inside $V_0$. The adjacency list of each vertex in the graph is a random permutation of its neighbor set. This completes the description of how a graph in this family is generated and presented to the algorithm.

The algorithm is given upfront the following information: $(i)$ the partition of the vertices into sets $V_0,V_1,\dots,V_{\sqrt{n}},V'_1,V'_2,\dots,V'_{\sqrt{n}}$, $(ii)$ degrees of all the vertices, and $(iii)$ all edges in the graph except the edges in the matching $M$. Thus the only task that remains for the algorithm is to discover enough information about the random matching $M$ so as to output a valid $(2\sqrt{n})$ coloring with probability at least $1 - o(1)$. This is the task we use to prove our lower bound.
The high level strategy is as follows. We first use an argument similar to the one in Section~\ref{sec:lower-matching} to argue that by making $o(n\sqrt{n})$, the algorithm is not able to find more than $o(n)$ edges of the matching $M$ (with constant probability). 
The algorithm now has
made all its queries and hence needs to commit to a coloring of the graph. We then show that no matter what coloring the algorithm chooses at this point, there is a non-trivial probability that one of the edges of $M$ not queried by the algorithm appears inside
one color class (i.e., becomes monochromatic), hence invalidating the output of the algorithm.

\begin{lemma} \label{lem:color-lb}
    Any algorithm does at most $n\sqrt{n}/400000$ queries on graphs generated by the distribution above, outputs a valid $(2\sqrt{n})$ coloring with probability at most $3/4$.
\end{lemma}

We adapt a few definitions from Section \ref{sec:lower-matching}. We say that an edge $(u,v) \in M$ has been {\em discovered} if the set of queries performed thus far uniquely identify the edge $(u, v)$ to be in $M$. 
For any discovered edge $(u,v) \in M$, we will say the vertices $u$ and $v$ have been {\em discovered};
any vertex in $V_0$ that has not been discovered is called an {\em undiscovered} vertex. 

After $t$ queries have been made by the algorithm, let $U(t) \subseteq V_0$ denote the set of undiscovered vertices in $V_0$. Let $E(t) \subseteq U(t) \times U(t)$ denote the set of edge slots that have not yet been queried/discovered. For any vertex $w \in V_0$, we say that {\em $Q_2$-uncertainty} of $w$ is $d$ if there are at least $d$ edges in $E(t)$ that are incident on $w$.

Additionally, we say that the state of the algorithm is {\em unsettled} after $t$ queries if there are at least $\frac{9n}{10}$ vertices in $U(t)$ whose $Q_2$-uncertainty is at least $\card{U(t)}-\frac{\sqrt{n}}{5000}$;  we will say that the state of the algorithm is {\em settled} otherwise.
The proof of Lemma \ref{lem:color-lb} has two parts. In the first part, we will show that if the state of the algorithm is unsettled after all the queries have been made, then any $(2\sqrt{n})$ coloring of the graph is invalid with some constant probability. In the second part, we will prove that to make the state of the algorithm settled, the algorithm needs $\Omega(n\sqrt{n})$ queries with a large constant probability.

The proof of the following lemma is analogous to the proof of Lemma \ref{lem:random_matching}. 

\begin{lemma} \label{lem:color-matching}
    Suppose we are given a graph $G$ on $n$ vertices such that each vertex in $G$ has at least $n-\sqrt{n}/4000$ neighbors. Then if we pick a perfect matching uniformly at random in $G$, for any edge $e$, the probability that $e$ is contained in the perfect matching is at most $\frac{1}{n-\sqrt{n}/1000}$. 
\end{lemma}

\begin{proof}
    For any edge $(u,v)$ in the graph, $u$ and $v$ have at least $n-\sqrt{n}/2000$ common neighbors. 
    Consider any perfect matching $M$ that contains the edge $(u,v)$. By the assumption on the degree of vertices in $G$, there are at least $n/2-\sqrt{n}/2000$ edges in the perfect matching $M$ such that both end-points of these edges are neighbors of $u$ and $v$. For each pair of such vertices $(a,b)$, we can then obtain two perfect matching by replacing $(u,v)$ and $(a,b)$ with $(u,a)$ and $(v,b)$ or $(u,b)$ and $(v,a)$. Thus for every matching $M$ containing the edge $(u,v)$, we can generate a unique set of $n-\sqrt{n}/1000$ perfect matchings that do not contain the edge $(u,v)$ (this is similar to the argument of Lemma \ref{lem:random_matching}). It then follows that the probability that a random perfect matching contains the edge $(u,v)$ is at most $\frac{1}{n-\sqrt{n}/1000}$
\end{proof}

We now prove that if the state of the algorithm is unsettled after it finishes the queries, then the coloring output by the algorithm is invalid with constant probability.

\begin{lemma} \label{lem:color-intra}
    If the state of the algorithm is unsettled after it finishes the queries, then any $(2\sqrt{n})$ coloring output by the algorithm is invalid with probability at least $3/8-o(1)$.
\end{lemma}

We start by the following key lemma. 

\begin{lemma} \label{lem:color-pairs}
    Given a graph $G$ on $n$ vertices such that each vertex in $G$ has at least $n-\sqrt{n}/4000$ neighbors. If we randomly pick a perfect matching uniformly, then for each vertex $v$, there are at most $\sqrt{n}/9$ vertices $u$ such that the probability that $(u,v)$ is contained in the perfect matching is less than $\frac{99}{100n}$. 
\end{lemma}

\begin{proof}
    Fix a vertex $v$, let $P_v(u)$ be the probability that the edge $(u,v)$ is contained in a random perfect matching, it is also the probability density function of the distribution of $v$'s neighbor in a random perfect matching. For any vertex $u$, by Lemma \ref{lem:color-matching}, $P_v(u) \le \frac{1}{n-\sqrt{n}/1000} < 1/n + \frac{1}{900n\sqrt{n}}$. Let $\mathcal{U}$ be the uniform distribution over the vertices of $G$, then the $\ell_1$-distance between the  two distributions satisfies $\norm{P_v-\mathcal{U}}_1 < \frac{1}{900\sqrt{n}}$. So there are $\leq \frac{\sqrt{n}}{9}$ vertices $u$ with $P_v(u)\le \frac{99}{100n}$.
\end{proof}

\begin{proof}[Proof of Lemma \ref{lem:color-intra}]
    Since the state of the algorithm is still unsettled, there are at most $\card{U(t)}-9n/10$ vertices whose $Q_2$-uncertainty is less than $\card{U(t)}-\sqrt{n}/5000$. We give the algorithm all the edges in $M$ incident on these vertices as well as a few additional edges in $M$ if needed  so that the number of undiscovered vertices becomes exactly $4n/5$. Let $U'$ be the set of these undiscovered vertices. After this step, these undiscovered vertices have $Q_2$-uncertainty at least $4n/5-\sqrt{n}/5000$. 

    Fix the output $2\sqrt{n}$ coloring of the graph by the algorithm. Let $S \subseteq U' \times U'$ be the set of pairs of vertices in $U'$ that have the same color. It is not hard to verify that $\card{S} \ge \frac{1}{2}\cdot \frac{4n}{5}\cdot (\frac{4n}{5}\cdot \frac{1}{2\sqrt{n}} - 1) \ge 0.15 n\sqrt{n}$. For any pair of vertices $(u,v)\in S$, let $X_{u,v}$ be the $0/1$ indicator variable that indicates $(u,v)$ is in $M$. Let $X=\sum_{(u,v)\in S} X_{u,v}$; then $\prob{X>0}$ is the probability that the coloring output by the algorithm is invalid. We will use Chebyshev's inequality to prove that this probability is a large enough constant.

    By Lemma \ref{lem:color-pairs}, for any vertex $v\in U'$, there are at most $\sqrt{n}/9$ vertices $u$ such that $\prob{X_{u,v}} < \frac{99}{100n}$. So 
    $$\expect{X} \ge (\card{S}-\card{U'}\cdot (\sqrt{n}/9))\cdot (99/100n) \ge (\card{S}-0.1n\sqrt{n})\cdot (99/100n) \ge 0.2\card{S}/n .$$
     
    Now consider any two pairs of vertices $u, v$ and $u', v'$ in $S$. If these two pairs share a vertex, then at least one of the pairs does not appear as an edge in the matching $M$, which means that the covariance of $X_{u,v}$ and $X_{u',v'}$ is negative. If they do not share a vertex, then $\prob{X_{u,v}\cdot X_{u',v'}=1}$ is the probability that both pairs are in $M$. If we pick a random matching conditioned on the event that $(u,v)$ is in the matching, using the same argument as in the proof of Lemma \ref{lem:color-matching}, the probability that $(u',v')$ is picked is at most $\frac{801}{800\card{U'}}$. So $\prob{X_{u',v'}=1 | X_{u,v}=1} \le \frac{801}{800\card{U'}}$, which means the $\cov{X_{u,v},X_{u',v'}} \leq \prob{X_{u,v}=1}\cdot (\frac{801}{800\card{U'}} - \prob{X_{u',v'}=1}) \le 0.015 /\card{U'}^2 \le 0.025/n^2$. So
    $$\var{X} \le \expect{X} + \sum_{(u,v),(u',v')\in S} 0.02/n^2 \le \expect{X} + 0.02\card{S}^2/n^2 \le \expect{X} + 5\expect{X}^2/8.$$
    By Chebyshev's inequality, $\prob{X=0} \le \frac{\var{X}}{\expect{X}^2} \le 5/8 + o(1)$.
\end{proof}

We next prove that to make the state of the algorithm settled, the algorithm needs $\Omega(n\sqrt{n})$ queries with large constant probability. 

\begin{lemma} \label{lem:color-query}
    If the algorithm only makes $n\sqrt{n}/400000$ queries, then with probability $0.9$, the state of the algorithm is unsettled.
\end{lemma}

\begin{proof}
    If the algorithm does a $Q_1$ query, say $Q_1(v,k)$, or does a $Q_2$ query, say $Q_2(u,v)$, then we say that the vertex $v$ has been queried. During the execution of the algorithm, once a vertex $v$ is queried $\sqrt{n}/5000$ times, we declare both the vertex $v$ and the vertex $u$ to which it is matched in $M$ as {\em bad} vertices, and any further queries on $v$ or $u$ are called {\em useless} queries. A query which is not useless is called a {\em useful} query. After all the queries, if an undiscovered vertex $v$ has not been queried $\sqrt{n}/5000$ times, then vertex $v$ has $Q_2$-uncertainty at least $U(t)-\sqrt{n}/5000$. So to prove the lemma, we need to prove that the total number of bad vertices and discovered vertices is at most $n/10$. The number of bad vertices is at most $n/20$ since at most two vertices are affected by a single query. On the other hand, any discovered vertex which is not bad must be discovered by a useful query. For a useful $Q_1$ query, say $Q_1(v,k)$, the probability that $v$ is discovered in this query is at most $\frac{1}{\sqrt{n}+1-\sqrt{n}/5000} \le \frac{2}{\sqrt{n}}$. For a useful $Q_2$ query, if the number of discovered vertices is less than $n/10$, then by Lemma \ref{lem:color-matching}, the probability that this query discover an edge in $M$ is at most $\frac{901\cdot 10}{900 \cdot 9 n} < \frac{1.2}{n}$. So for any useful query, the probability that it discovers an edge in $M$ is at most $\frac{2}{\sqrt{n}}$. As such, the expected number of useful queries which discover an edge in $M$ is at most $n/200000$.
 By Markov, this implies that the number of useful queries that discover an edge $M$ is with probability $0.9$ at most $n/20000$.  
 So the number of discovered vertices which are not bad is at most $n/10000$. This in turn implies that the total number of bad vertices and discovered vertices is less than $n/10$ with probability $0.9$.
\end{proof}

\begin{proof} [Proof of Lemma \ref{lem:color-lb}]
    By Lemma \ref{lem:color-intra}, if the algorithm only makes $n\sqrt{n}/400000$ queries, then with probability $0.9$ the state of the algorithm is unsettled. Conditioned on this event, by Lemma \ref{lem:color-query}, the output of the algorithm is an invalid coloring with probability at least $3/8-o(1)$. So with probability $3/8-0.1-o(1) \geq 1/4$, the output of any algorithm that makes less than $n\sqrt{n}/400000$ queries is an invalid coloring.
\end{proof}

Theorem \ref{thm:lower-color} is directly implied by Lemma \ref{lem:color-lb} (by modifying constants to allow for $c\sqrt{n}$ coloring as opposed to $2\sqrt{n}$).

\subsection*{Acknowledgements}

The first author would like to thank Merav Parter for helpful comments and the pointer to~\cite{ParterS18}, and T.S. Jayram and Saeed Seddighin for valuable discussions regarding the applicability of our palette-sparsification theorem to $c$-coloring 
problem for other choices of $c$ that prompted us to include Appendix~\ref{app:sparsification-general}. We are also grateful to anonymous reviewers of SODA'19 for many valuable comments and suggestions.

\bibliographystyle{abbrv}
\bibliography{general}

\begin{thebibliography}{10}

\bibitem{Ablayev93}
F.~M. Ablayev.
\newblock Lower bounds for one-way probabilistic communication complexity.
\newblock In {\em Automata, Languages and Programming, 20nd International
  Colloquium, ICALP93, Lund, Sweden, July 5-9, 1993, Proceedings}, pages
  241--252, 1993.

\bibitem{AhnG15}
K.~J. Ahn and S.~Guha.
\newblock Access to data and number of iterations: Dual primal algorithms for
  maximum matching under resource constraints.
\newblock In {\em Proceedings of the 27th {ACM} on Symposium on Parallelism in
  Algorithms and Architectures, {SPAA} 2015, Portland, OR, USA, June 13-15,
  2015}, pages 202--211, 2015.

\bibitem{AhnGM12Linear}
K.~J. Ahn, S.~Guha, and A.~McGregor.
\newblock Analyzing graph structure via linear measurements.
\newblock In {\em Proceedings of the Twenty-third Annual ACM-SIAM Symposium on
  Discrete Algorithms}, SODA '12, pages 459--467. SIAM, 2012.

\bibitem{alon1995note}
N.~Alon.
\newblock A note on network reliability.
\newblock In {\em Discrete Probability and Algorithms}, pages 11--14. Springer,
  1995.

\bibitem{AlonRVX12}
N.~Alon, R.~Rubinfeld, S.~Vardi, and N.~Xie.
\newblock Space-efficient local computation algorithms.
\newblock In {\em Proceedings of the Twenty-Third Annual {ACM-SIAM} Symposium
  on Discrete Algorithms, {SODA} 2012, Kyoto, Japan, January 17-19, 2012},
  pages 1132--1139, 2012.

\bibitem{AndoniNOY14}
A.~Andoni, A.~Nikolov, K.~Onak, and G.~Yaroslavtsev.
\newblock Parallel algorithms for geometric graph problems.
\newblock In {\em Symposium on Theory of Computing, {STOC} 2014, New York, NY,
  USA, May 31 - June 03, 2014}, pages 574--583, 2014.

\bibitem{Assadi17}
S.~Assadi.
\newblock Simple round compression for parallel vertex cover.
\newblock {\em CoRR}, abs/1709.04599, 2017.

\bibitem{AssadiBBMS18}
S.~Assadi, M.~Bateni, A.~Bernstein, V.~S. Mirrokni, and C.~Stein.
\newblock Coresets meet {EDCS:} algorithms for matching and vertex cover on
  massive graphs.
\newblock In {\em Proceedings of the 30th Annual {ACM-SIAM} Symposium on
  Discrete Algorithms, {SODA} 2019}, 2019.

\bibitem{AssadiKLY16}
S.~Assadi, S.~Khanna, Y.~Li, and G.~Yaroslavtsev.
\newblock Maximum matchings in dynamic graph streams and the simultaneous
  communication model.
\newblock In {\em Proceedings of the Twenty-Seventh Annual {ACM-SIAM} Symposium
  on Discrete Algorithms, {SODA} 2016, Arlington, VA, USA, January 10-12,
  2016}, pages 1345--1364, 2016.

\bibitem{BarenboimE11}
L.~Barenboim and M.~Elkin.
\newblock Distributed deterministic edge coloring using bounded neighborhood
  independence.
\newblock In {\em Proceedings of the 30th Annual {ACM} Symposium on Principles
  of Distributed Computing, {PODC} 2011, San Jose, CA, USA, June 6-8, 2011},
  pages 129--138, 2011.

\bibitem{BeameKS13}
P.~Beame, P.~Koutris, and D.~Suciu.
\newblock Communication steps for parallel query processing.
\newblock In {\em Proceedings of the 32nd {ACM} {SIGMOD-SIGACT-SIGART}
  Symposium on Principles of Database Systems, {PODS} 2013, New York, NY, {USA}
  - June 22 - 27, 2013}, pages 273--284, 2013.

\bibitem{BehnezhadDH18}
S.~Behnezhad, M.~Derakhshan, and M.~Hajiaghayi.
\newblock Brief announcement: Semi-mapreduce meets congested clique.
\newblock {\em CoRR}, abs/1802.10297, 2018.

\bibitem{BeraG18}
S.~K. Bera and P.~Ghosh.
\newblock Coloring in graph streams.
\newblock {\em CoRR}, abs/1807.07640, 2018.

\bibitem{ChangLP18}
Y.~Chang, W.~Li, and S.~Pettie.
\newblock An optimal distributed ({\(\Delta\)}+1)-coloring algorithm?
\newblock In {\em Proceedings of the 50th Annual {ACM} {SIGACT} Symposium on
  Theory of Computing, {STOC} 2018, Los Angeles, CA, USA, June 25-29, 2018},
  pages 445--456, 2018.

\bibitem{ChangFGUZ18}
Y.-J. Chang, M.~Fischer, M.~Ghaffari, J.~Uitto, and Y.~Zheng.
\newblock The complexity of {({$\Delta$}+1)} coloring in congested clique,
  massively parallel computation, and centralized local computation.
\newblock {\em CoRR}, abs/1808.08419, 2018.

\bibitem{CormodeDK18}
G.~Cormode, J.~Dark, and C.~Konrad.
\newblock Independent sets in vertex-arrival streams.
\newblock {\em CoRR}, abs/1807.08331, 2018.

\bibitem{CzumajLMMOS18}
A.~Czumaj, J.~Lacki, A.~Madry, S.~Mitrovic, K.~Onak, and P.~Sankowski.
\newblock Round compression for parallel matching algorithms.
\newblock In {\em Proceedings of the 50th Annual {ACM} {SIGACT} Symposium on
  Theory of Computing, {STOC} 2018, Los Angeles, CA, USA, June 25-29, 2018},
  pages 471--484, 2018.

\bibitem{ElkinPS15}
M.~Elkin, S.~Pettie, and H.~Su.
\newblock (2{\(\Delta\)} - l)-edge-coloring is much easier than maximal
  matching in the distributed setting.
\newblock In {\em Proceedings of the Twenty-Sixth Annual {ACM-SIAM} Symposium
  on Discrete Algorithms, {SODA} 2015, San Diego, CA, USA, January 4-6, 2015},
  pages 355--370, 2015.

\bibitem{FKMSZ05}
J.~Feigenbaum, S.~Kannan, A.~McGregor, S.~Suri, and J.~Zhang.
\newblock On graph problems in a semi-streaming model.
\newblock {\em Theor. Comput. Sci.}, 348(2-3):207--216, 2005.

\bibitem{FischerGK17}
M.~Fischer, M.~Ghaffari, and F.~Kuhn.
\newblock Deterministic distributed edge-coloring via hypergraph maximal
  matching.
\newblock In {\em 58th {IEEE} Annual Symposium on Foundations of Computer
  Science, {FOCS} 2017, Berkeley, CA, USA, October 15-17, 2017}, pages
  180--191, 2017.

\bibitem{FrahlingIS2008}
G.~Frahling, P.~Indyk, and C.~Sohler.
\newblock Sampling in dynamic data streams and applications.
\newblock {\em International Journal of Computational Geometry \&
  Applications}, 18(01n02):3--28, 2008.

\bibitem{GhaffariGMR18}
M.~Ghaffari, T.~Gouleakis, C.~Konrad, S.~Mitrovic, and R.~Rubinfeld.
\newblock Improved massively parallel computation algorithms for mis, matching,
  and vertex cover.
\newblock In {\em Proceedings of the 2018 {ACM} Symposium on Principles of
  Distributed Computing, {PODC} 2018, Egham, United Kingdom, July 23-27, 2018},
  pages 129--138, 2018.

\bibitem{Goldreich17}
O.~Goldreich.
\newblock {\em Introduction to Property Testing}.
\newblock Cambridge University Press, 2017.

\bibitem{GoodrichSZ11}
M.~T. Goodrich, N.~Sitchinava, and Q.~Zhang.
\newblock Sorting, searching, and simulation in the mapreduce framework.
\newblock In {\em Algorithms and Computation - 22nd International Symposium,
  {ISAAC} 2011, Yokohama, Japan, December 5-8, 2011. Proceedings}, pages
  374--383, 2011.

\bibitem{HarrisSS16}
D.~G. Harris, J.~Schneider, and H.-H. Su.
\newblock Distributed {({\(\Delta\)}+ 1)}-coloring in sublogarithmic rounds.
\newblock In {\em Proceedings of the forty-eighth annual ACM symposium on
  Theory of Computing}, pages 465--478. ACM, 2016.

\bibitem{HarveyLL18}
N.~J.~A. Harvey, C.~Liaw, and P.~Liu.
\newblock Greedy and local ratio algorithms in the mapreduce model.
\newblock In {\em Proceedings of the 30th on Symposium on Parallelism in
  Algorithms and Architectures, {SPAA} 2018, Vienna, Austria, July 16-18,
  2018}, pages 43--52, 2018.

\bibitem{HopcroftK71}
J.~E. Hopcroft and R.~M. Karp.
\newblock A n{\^{}}5/2 algorithm for maximum matchings in bipartite graphs.
\newblock In {\em 12th Annual Symposium on Switching and Automata Theory, East
  Lansing, Michigan, USA, October 13-15, 1971}, pages 122--125, 1971.

\bibitem{IndykMMM14}
P.~Indyk, S.~Mahabadi, M.~Mahdian, and V.~S. Mirrokni.
\newblock Composable core-sets for diversity and coverage maximization.
\newblock In {\em Proceedings of the 33rd {ACM} {SIGMOD-SIGACT-SIGART}
  Symposium on Principles of Database Systems, PODS'14, Snowbird, UT, USA, June
  22-27, 2014}, pages 100--108, 2014.

\bibitem{JowhariST2011}
H.~Jowhari, M.~Sa{\u{g}}lam, and G.~Tardos.
\newblock Tight bounds for lp samplers, finding duplicates in streams, and
  related problems.
\newblock In {\em Proceedings of the thirtieth ACM SIGMOD-SIGACT-SIGART
  symposium on Principles of database systems}, pages 49--58. ACM, 2011.

\bibitem{KapralovNPWWY17}
M.~Kapralov, J.~Nelson, J.~Pachocki, Z.~Wang, D.~P. Woodruff, and
  M.~Yahyazadeh.
\newblock Optimal lower bounds for universal relation, and for samplers and
  finding duplicates in streams.
\newblock In {\em 58th {IEEE} Annual Symposium on Foundations of Computer
  Science, {FOCS} 2017, Berkeley, CA, USA, October 15-17, 2017}, pages
  475--486, 2017.

\bibitem{KarloffSV10}
H.~J. Karloff, S.~Suri, and S.~Vassilvitskii.
\newblock A model of computation for mapreduce.
\newblock In {\em Proceedings of the Twenty-First Annual {ACM-SIAM} Symposium
  on Discrete Algorithms, {SODA} 2010, Austin, Texas, USA, January 17-19,
  2010}, pages 938--948, 2010.

\bibitem{Konrad18}
C.~Konrad.
\newblock {MIS} in the congested clique model in {$O$(log log {\(\Delta\)})}
  rounds.
\newblock {\em CoRR}, abs/1802.07647, 2018.

\bibitem{KremerNR95}
I.~Kremer, N.~Nisan, and D.~Ron.
\newblock On randomized one-round communication complexity.
\newblock In {\em Proceedings of the Twenty-Seventh Annual {ACM} Symposium on
  Theory of Computing, 29 May-1 June 1995, Las Vegas, Nevada, {USA}}, pages
  596--605, 1995.

\bibitem{KumarMVV13}
R.~Kumar, B.~Moseley, S.~Vassilvitskii, and A.~Vattani.
\newblock Fast greedy algorithms in mapreduce and streaming.
\newblock In {\em 25th {ACM} Symposium on Parallelism in Algorithms and
  Architectures, {SPAA} '13, Montreal, QC, Canada - July 23 - 25, 2013}, pages
  1--10, 2013.

\bibitem{LattanziMSV11}
S.~Lattanzi, B.~Moseley, S.~Suri, and S.~Vassilvitskii.
\newblock Filtering: a method for solving graph problems in mapreduce.
\newblock In {\em {SPAA} 2011: Proceedings of the 23rd Annual {ACM} Symposium
  on Parallelism in Algorithms and Architectures, San Jose, CA, USA, June 4-6,
  2011 (Co-located with {FCRC} 2011)}, pages 85--94, 2011.

\bibitem{Luby85}
M.~Luby.
\newblock A simple parallel algorithm for the maximal independent set problem.
\newblock In {\em Proceedings of the 17th Annual {ACM} Symposium on Theory of
  Computing, May 6-8, 1985, Providence, Rhode Island, {USA}}, pages 1--10,
  1985.

\bibitem{Luby88}
M.~Luby.
\newblock Removing randomness in parallel computation without a processor
  penalty.
\newblock In {\em 29th Annual Symposium on Foundations of Computer Science,
  White Plains, New York, USA, 24-26 October 1988}, pages 162--173, 1988.

\bibitem{McGregorTVV15}
A.~McGregor, D.~Tench, S.~Vorotnikova, and H.~T. Vu.
\newblock Densest subgraph in dynamic graph streams.
\newblock In {\em Mathematical Foundations of Computer Science 2015 - 40th
  International Symposium, {MFCS} 2015, Milan, Italy, August 24-28, 2015,
  Proceedings, Part {II}}, pages 472--482, 2015.

\bibitem{NaorS93}
M.~Naor and L.~J. Stockmeyer.
\newblock What can be computed locally?
\newblock In {\em Proceedings of the Twenty-Fifth Annual {ACM} Symposium on
  Theory of Computing, May 16-18, 1993, San Diego, CA, {USA}}, pages 184--193,
  1993.

\bibitem{NguyenO08}
H.~N. Nguyen and K.~Onak.
\newblock Constant-time approximation algorithms via local improvements.
\newblock In {\em 49th Annual {IEEE} Symposium on Foundations of Computer
  Science, {FOCS} 2008, October 25-28, 2008, Philadelphia, PA, {USA}}, pages
  327--336, 2008.

\bibitem{PanconesiR01}
A.~Panconesi and R.~Rizzi.
\newblock Some simple distributed algorithms for sparse networks.
\newblock {\em Distributed Computing}, 14(2):97--100, 2001.

\bibitem{Parter18}
M.~Parter.
\newblock ({$\Delta$}+1) coloring in the congested clique model.
\newblock In {\em 45th International Colloquium on Automata, Languages, and
  Programming, {ICALP} 2018, July 9-13, 2018, Prague, Czech Republic}, pages
  160:1--160:14, 2018.

\bibitem{ParterS18}
M.~Parter and H.~Su.
\newblock Randomized {($\Delta$+1)}-coloring in {$O(\log^*{\Delta})$} congested
  clique rounds.
\newblock In {\em 32nd International Symposium on Distributed Computing, {DISC}
  2018, New Orleans, LA, USA, October 15-19, 2018}, pages 39:1--39:18, 2018.

\bibitem{SchneiderW10}
J.~Schneider and R.~Wattenhofer.
\newblock A new technique for distributed symmetry breaking.
\newblock In {\em Proceedings of the 29th Annual {ACM} Symposium on Principles
  of Distributed Computing, {PODC} 2010, Zurich, Switzerland, July 25-28,
  2010}, pages 257--266, 2010.

\bibitem{Yao87}
A.~C. Yao.
\newblock Lower bounds to randomized algorithms for graph properties (extended
  abstract).
\newblock In {\em 28th Annual Symposium on Foundations of Computer Science, Los
  Angeles, California, USA, 27-29 October 1987}, pages 393--400, 1987.

\end{thebibliography}

\clearpage
\appendix

\section{Proof of Lemma~\ref{lem:sparse-color}: Coloring Sparse Vertices}\label{app:sparse}

Recall the definition of sparse vertices $\Vsparse_\star$ in the extended HSS-decomposition from Section~\ref{sec:prelim}. We now show that we can color all sparse vertices using only the colors in the first batch. 

\begin{lemma*}[Restatement of Lemma~\ref{lem:sparse-color}]
    With high probability, there exists a partial coloring function $\col_1:V \rightarrow [\Delta+1] \cup \set{\perp}$ such that for all vertices $v \in \Vsparse_\star$, $\col_1(v) \in L_1(v)$.  In other words, 
    one can assign a valid color to all sparse vertices using only the colors in the first batch. 
\end{lemma*}

Firstly, since every vertex samples each color with probability $p := \Paren{\frac{\alpha \cdot \log{n}}{3\eps^2 \cdot (\Delta+1)}}$, we have that with high probability (by Chernoff bound), the number 
of sampled colors per vertex in $L_1(v)$ is at least $K := (\alpha/6) \cdot \log{n}/\eps^2$. In the following, we condition on this high probability event and simply ignore the remaining (if any) colors of vertices. 
This allows us to assume that each vertex samples $K$ colors uniformly at random from $[\Delta+1]$ in $L_1(v)$. To be more formal, we let $L_1(v):= \set{c_1(v),\ldots,c_K(v)}$ where each $c_i(v)$ is chosen
uniformly at random from $[\Delta+1]$.

We prove this lemma by showing that one can simulate a natural greedy algorithm (similar but not identical to the algorithm of~\cite{ElkinPS15}) for coloring sparse vertices using only the colors in the first batch.  
The first step is to use the \emph{first} color in $L_1(v)$, i.e., $c_1(v)$ chosen uniformly at random from $[\Delta+1]$, 
for all vertices $v \in \Vsparse_\star$ to color a large fraction of vertices in $\Vsparse_\star$; the main property of this coloring is that any uncolored 
$\eps$-sparse vertex has $\Omega(\eps^2\Delta)$ ``excess'' colors compared to the number of edges it has to other \emph{uncolored} $\eps$-sparse vertices. 
This step follows from the proof of the \textsf{OneShotColoring} procedure in~\cite{ElkinPS15,HarrisSS16,ChangLP18} and we simply present a proof sketch for intuition.  
We then use the remaining colors in $L_1(\cdot)$ for each uncolored vertex and color them greedily, using the fact that the number of available colors is sufficiently larger than the number of 
neighbors of each uncolored vertex in every step. This part is also similar to the algorithm in~\cite{ElkinPS15} (see also~\cite{HarrisSS16,ChangLP18}) 
but uses a slightly different argument as here we \emph{cannot} sample the colors for each vertex \emph{adaptively} (as the colors in
$L_1(\cdot)$ are  sampled ``at the beginning'' of the greedy algorithm).

\subsubsection*{Creating Excess Colors for Vertices}

Consider the following process: For each vertex $v \in \Vsparse_\star$, let $x(v)$ denote the \emph{first} color in $L_1(v)$, i.e., $x(v) = c_1(v)$ (which is chosen uniformly at random from $[\Delta+1]$); 
we set $\col_1(v) = x(v)$ iff for every vertex $u \in N(v)$, $x(v) \neq x(u)$. 
Let $\Vr$ be the remaining set of uncolored $\eps$-sparse vertices and $\Gr$ be the \emph{induced} subgraph of $G$ on $\Vr$. Finally, define $\drem{v}$ as the degree of $v$ in the graph $\Gr$. 
We have the following key lemma (recall that $A_\col(\cdot)$ denotes the set of available colors under the partial coloring $\col$). 
 
\begin{lemma}[cf.~\cite{ElkinPS15,HarrisSS16,ChangLP18}]\label{lem:one-shot}
	 For any vertex $v \in \Vr$, $\card{A_{\col_1}(v)} \geq \drem{v} + e^{-6} \cdot \eps^2\cdot\Delta$ with probability at least $1-\exp\paren{-\Omega(\eps^2 \cdot \Delta)}$. 
\end{lemma}
\begin{proof}[Proof Sketch]
	Fix any vertex $v \in \Vsparse_\star$. We assume that degree of $v$ is exactly $\Delta$; if not, we can simply assign arbitrary vertices as neighbors for $v$ until its degree reaches $\Delta$ and require that $v$ does not share the same color
	with these vertices (it would be evident from the proof below that this is without loss of generality). Let $N(v) := \set{u_1,\ldots,u_\Delta}$ be the neighborhood of $v$ in $G$. Recall that in the extended HSS-decomposition, 
	by Lemma~\ref{lem:extended-HSS-decomposition}, $v$ is $\eps$-sparse and hence by Proposition~\ref{prop:sparse-vertices}, there are at least $t:= \eps^2 \cdot {{\Delta}\choose{2}}$ \emph{non-edges} in $N(v)$. 
	Let $f_1,\ldots,f_t$ be these non-edges. 
	
	Consider an arbitrary ordering of vertices. 	We say that a non-edge $f_i = (u_i,v_i)$ for $i \in [t]$ is \emph{good} iff $(i)$ $x(u_i) = x(v_i)$, and $(ii)$ no neighbor $w$ of $u_i$ or $v_i$ has $x(w) = x(u_i) (=x(v_i))$. Let $S_v$ be the set of 
	all good non-edges. Notice that both vertices $u_i$ and $v_i$ for a good edge $(u_i,v_i)$ would be assigned the \emph{same} color $x(u_i)=x(v_i)$ by $\col_1$. As such, the number of available colors for $v$ under $c_1$ is at 
	least $A_{\col_1}(v) \geq (\Delta+1) - (\Delta - \drem{v}) + \card{S_v} = \drem{v} + \card{S_v} + 1$. It thus remains to bound the size of $S_v$. Define indicator random variables $Y_i$ for $i \in [t]$, where
	$Y_i = 1$  iff the non-edge $f_i$ is good. Let $Y:=\sum_{i}Y_i$. We thus have $\card{S_v} = Y$. We bound the expectation of $Y$ below.
	
	\begin{claim}\label{clm:S_v-exp}
		$\Ex\bracket{Y} \geq e^{-5} \cdot \eps^2 \Delta$. 
	\end{claim}
	\begin{proof}
	 We have, 
		\begin{align*}
			\Pr\bracket{\text{$Y_i$ is good}} &= \Pr\Bracket{\paren{x(u_i) = x(v_i) }\wedge \paren{\forall w \in N(u_i) \cup N(v_i): x(w) \neq x(u_i)}} \\
			&\geq \frac{1}{\Delta+1} \cdot \paren{1-\frac{1}{\Delta+1}}^{\card{N(u_i) \cup N(v_i)}} \geq \frac{1}{\Delta+1} \cdot \paren{1-\frac{1}{\Delta+1}}^{2\Delta} \geq \frac{e^{-4}}{\Delta}.  \tag{as $1-x \leq e^{-x} \leq 1-\frac{x}{2}$ for $x \in (0,1)$}
		\end{align*} 
		By linearity of expectation, $\Ex\bracket{Y} \geq t \cdot \frac{e^{-4}}{\Delta} \geq e^{-5} \cdot \eps^2 \Delta$, using the bound on the value of $t$. 
	\end{proof}
	
	The next step is to prove a concentration bound for $Y$. The challenge here is that variables $Y_1,\ldots,Y_{t}$ are not independent of each other and one cannot readily use Chernoff bound. Still, one can 
	define an appropriate martingale sequence to prove that: 
	\begin{align*}
		\Pr\bracket{Y \leq e^{-1} \cdot \Ex\bracket{Y}} \leq \exp\paren{-\Omega(\eps^2\cdot\Delta)}.
	\end{align*}
	As this is not the contribution of our paper, we omit the details and instead refer the reader to the proofs in~\cite{ElkinPS15,HarrisSS16,ChangLP18} (see Lemma~3.1 in~\cite{ElkinPS15}, Lemma~5.5 in~\cite{HarrisSS16}, or Lemma~5 in~\cite{ChangLP18}).
	
	To conclude, with probability $1-\exp\paren{-\Omega(\eps^2 \cdot \Delta)}$, $\card{S_v} \geq e^{-6}\cdot\eps^2\Delta$. As argued before, this event implies that $\card{A_{\col_1}(v)} \geq \drem{v} + \card{S_v} \geq \drem{v} + \eps^{-6} \cdot \eps^2 \cdot \Delta$, 
	finalizing the proof.
\end{proof}

\subsubsection*{Coloring the Remaining Sparse Vertices}

We now color the remaining vertices, i.e., the vertices in $\Gr$, using the remaining $K-1$ colors in batch one of every vertex $v$. Recall that by the lower bound of $\Omega(\log{n})$ on the value of $\Delta$, 
and by Lemma~\ref{lem:one-shot}, for any vertex in $\Gr$, we have that $\card{A_{c_1}(v)} \geq \drem{v} + e^{-6} \cdot \eps^2\cdot\Delta$ with high probability. We take a union bound over all vertices and in the following assume that
this property holds for every vertex. Consider the following greedy process: 

\begin{tbox}
\GC: The greedy procedure for coloring $\Gr$ using the first batch of colors. 
\begin{enumerate}
	\item For $i=2$ to $K$: 
	\begin{enumerate}
		\item For any vertex $v \in \Vr$, let $x_i(v)$ be the $i$-th color in $L_1(v)$, i.e., $x_i(v) = c_i(v)$. 
		\item If $x_i(v) \in A_{\col_1}(v)$ and no vertex $u$ in $N_{\Gr}(v)$ has also $x_i(u) = x_i(v)$, let $\col_1(v) = x_i(v)$. 
	\end{enumerate}
\end{enumerate}
\end{tbox}

It is clear that $\GC$ only uses the colors in the first batch to color the vertices. We argue that after running $\GC$, all vertices in $\Gr$ are colored by $\col_1$ with high probability. 

\begin{lemma}\label{lem:gc-colors}
	After running $\GC$, any vertex $v \in \Vr$ receives a valid color by $\col_1$ with high probability. 
\end{lemma}
\begin{proof}
	Fix any vertex $v \in \Vr$. We argue that at every iteration of the for-loop in $\GC$, there is a constant probability that we assign a valid color to $v$.  As we never change the color of a vertex once it is assigned a valid color by $\col_1$, this means that
	after the $K-1 = \Omega(\log{n})$ iterations, $v$ is going to receive a valid color with high probability. Taking a union bound over all vertices then finalizes the proof.
	
	We now prove the bound on the probability of obtaining a valid coloring in each iteration $i \in [2:K]$. Recall that at the beginning of running $\GC$, $\card{A_{\col_1}(v)} \geq \drem{v} + e^{-6} \cdot \eps^2\cdot\Delta$. 
	Consider iteration $i$ of the for-loop and let $C$ denote the set of colors in $\bigcup_{u \in N_{\Gr}(v)} \col_1(u) \cup x_i(u)$ (here $N_{\Gr}(v)$ is the set of neighbors of $v$ in $\Gr$). 
	Note that $\card{C} \leq \drem{v}$. As such, $A_{\col_1}(v) \setminus C$ has 
	size at least $e^{-6} \cdot \eps^2\cdot\Delta$. If $x_i(v)$ belongs to this set, then $\col_1$ is going to assign this color permanently to $v$. Moreover, recall that $x_i(v)$ is chosen uniformly at random from $[\Delta+1]$.
	As such, 
	\begin{align*}
		\Pr\Bracket{\textnormal{$\col_1$ assigns a valid color to $v$ at iteration $i$}} = \Pr\Bracket{x_i(v) \in A_{\col_1}(v) \setminus C} \geq e^{-6}\cdot\eps^2 \cdot (\frac{\Delta}{\Delta+1}),
	\end{align*}
	which is $\Omega(1)$. 
	After repeating this process $K-1$ times independently (as the choice of colors in $L_1(v)$ are independent), the probability that $v$ is not colored by $\col_1$ is at most
	\begin{align*}
		\paren{1-e^{-6}\cdot\eps^2 \cdot (\frac{\Delta}{\Delta+1})}^{K-1} \leq \exp\Paren{e^{-6.1}\cdot\eps^2\cdot(\alpha/6)\log{n}/\eps^2} \leq \exp\Paren{e^{-6.1}\cdot(\alpha/6)\log{n}} \leq \frac{1}{n^{50}} \tag{by the choice of $\alpha$}.
	\end{align*}
	Taking a union bound over $n$ vertices now finalizes the proof.
\end{proof}

\noindent
Lemma~\ref{lem:sparse-color} now follows immediately from Lemmas~\ref{lem:one-shot} and~\ref{lem:gc-colors}.

\section{Proof of Claim~\ref{lem:sumvar}}\label{app:sumvar}
\begin{claim*}[Restatement of Claim~\ref{lem:sumvar}]
    Suppose $0<a<1$ is a constant. Consider two random variables $X:=\sum_{i=1}^{n}X_i$ and $Y:=\sum_{i=1}^{n}Y_i$ where for all $i \in [n]$, $X_i$ and $Y_i$ are independent indicator random variables and $\prob{X_i=1}=a^{k_i}$ and $\prob{Y_i=1}=a^{\ell_i}$. Suppose $k_i$ and $\ell_i$ are non-negative integers that are indexed in decreasing order and have the following two properties:
    \begin{itemize}
        \item for any $j \le n$, $\sum_{i=1}^j k_i \le \sum_{i=1}^j \ell_i$
        \item $\sum_{i=1}^n k_i = \sum_{i=1}^n \ell_i$
    \end{itemize}
    then for any integer $M$, $\prob{X\ge M}\le \prob{Y\ge M}$.
\end{claim*}

\begin{proof}
    We construct a series of $N$ (to be determined later) random variables $X=Z^0,Z^1,\dots,Z^N=Y$ as follows: Each $Z^i$ is the sum of $n$ independent indicator random variable $Z^i_1,Z^i_2,\dots,Z^i_n$ where
     $\prob{Z^i_j=1}=a^{k^i_j}$ (for $k^i_j$ determined later). When constructing $Z^{i+1}$, let $j_1$ be the first index such that $k^i_{j_1}<\ell_{j_1}$, $j_2$ be the  first index such that $k^i_{j_2}>\ell_{j_2}$. Then $k^{i+1}_{j_1}=k^i_{j_1}+1$, $k^{i+1}_{j_2}=k^i_{j_2}-1$, and $k^{i+1}_j=k^i_j$ for all other $j$.

    Since $k_i$ and $\ell_i$ are integers, and $\sum_{i=1}^n k_i = \sum_{i=1}^n \ell_i$, the construction will terminate. On the other hand, since for any $j\le n$, $\sum_{i=1}^j k_i \le \sum_{i=1}^j \ell_i$, each time when we construct the next random variable, $j_1<j_2$, in other word, $k^i_{j_1}\ge k^i_{j_2}$ for any $0 \le i < N$. We will prove that for each $0\le i < N$, $\prob{Z^i>M} \le \prob{Z^{i+1}>M}$, which implies the lemma.

    Fix an $i$. For any $j$ which is not $j_1$ and $j_2$, random variables $Z^i_j$ and $Z^{i+1}_j$ have the same distribution. Thus by coupling, we only need to prove that for any integer $M'$, $\prob{Z^i_{j_1}+Z^i_{j_2} \ge M'} \le \prob{Z^{i+1}_{j_1}+Z^{i+1}_{j_2} \ge M'}$. Since the sum of two indicator random variables is always non-negative and is at most $2$, we only need to consider the cases when $M'$ is $1$ or $2$. When $M'=2$, by the fact that $k^i_{j_1}+k^i_{j_2}=k^{i+1}_{j_1}+k^{i+1}_{j_2}$,
    $$\prob{Z^i_{j_1}+Z^i_{j_2} \ge 2} = a^{k^i_{j_1}+k^i_{j_2}} = a^{k^{i+1}_{j_1}+k^{i+1}_{j_2}} = \prob{Z^{i+1}_{j_1}+Z^{i+1}_{j_2} \ge 2}.$$
    When $M'=1$, by the fact that $k^i_{j_1} \ge k^i_{j_2}$,
    $$\prob{Z^i_{j_1}+Z^i_{j_2} \ge 1} = a^{k^i_{j_1}} + a^{k^i_{j_2}} - a^{k^i_{j_1}+k^i_{j_2}} \le a^{k^{i+1}_{j_1}} + a^{k^{i+1}_{j_2}} - a^{k^{i+1}_{j_1}+k^{i+1}_{j_2}} = \prob{Z^{i+1}_{j_1}+Z^{i+1}_{j_2} \ge 1}. $$
    This concludes the proof.
\end{proof}

\section{Further Remarks on the Palette-Sparsification Theorem}\label{app:sparsification} 

We present some further basic observations on our palette-sparsification theorem (Theorem~\ref{thm:color-sampling}) in this section. In the following, we first show that the size of $O(\log{n})$ colors is optimal 
for this theorem even when $\Delta$ is much smaller, and then show that this theorem is inherently suitable only for $(\Delta+1)$ coloring problem and not even $\Delta$-coloring or any smaller number of colors. 

\subsection{Palette-Sparsification with Smaller List Size}\label{app:logn-optimal}

We prove that the choice of $O(\log{n})$ colors in Theorem~\ref{thm:color-sampling} is essentially optimal. That is, if we instead sample $o(\frac{\log{n}}{\log\log{n}})$ colors per each vertex, then
there exists graphs such that the resulting list-coloring instance has no valid solution. 

\begin{proposition}\label{prop:logn-optimal}
  There exists an $n$-vertex graph $G$ with maximum degree $\Delta$ such that if for each vertex $v \in V$,
    we independently pick a set $L(v)$ of colors with size $o(\frac{\log{n}}{\log\log{n}})$ uniformly at random from $[\Delta+1]$, then, 
    with probability $1-o(1)$, there exists no proper coloring $\col: V \rightarrow [\Delta+1]$ of $G$ such that for all vertices $v \in V$, $\col(v) \in L(v)$. 
\end{proposition}

\begin{proof}
	Let $K = o(\frac{\log{n}}{\log\log{n}})$ be the number of sampled colors in $L(v)$ for each $v \in V$. Consider a graph $G$ which is a collection of $(K+2)$-cliques $C_1,\ldots,C_k$ for $k = n/(K+1)$. As such, maximum degree
	of this graph is $\Delta = K+1$. 	For $G$ to be list-colorable with lists $L(v)$, in each clique $C_i$, we need to sample all the $(\Delta+1)$ colors; in particular, color $1$ needs to be sampled by at least one vertex in $C_i$. 	
	We say that the clique $C_i$ is \emph{missing $1$} iff this color is not sampled by any vertex in $C_i$. We have,
	\begin{align*}
		\Pr\bracket{\textnormal{$C_i$ is missing $1$}} = \Pr\bracket{\forall v \in C_i: 1 \notin L(v)} = \prod_{v \in C_i} \Pr\bracket{1 \notin L(v)} = \Paren{\frac{{{K+1}\choose{K}}}{{{K+2}\choose{K}}}}^{K+2} = \Paren{\frac{1}{K+2}}^{K+2}.
	\end{align*}
	On the other hand: 
	\begin{align*}
		\Pr\bracket{\exists i \in [k]: \textnormal{$C_i$ is missing $1$}} = 1-\Paren{1-\paren{\frac{1}{K+2}}^{K+2}}^{k} \geq 1-\exp\Paren{-\frac{n}{K+1} \cdot \paren{\frac{1}{K+2}}^{K+2}}.
	\end{align*}
	As $K = o(\frac{\log{n}}{\log\log{n}})$, $\frac{n}{K+1} = \omega\Paren{\paren{K+2}^{K+2}}$, and hence we have, $\Pr\bracket{\exists i \in [k]: \textnormal{$C_i$ is missing $1$}} = 1-o(1)$. As such, with sufficiently large constant probability,
	list-coloring of $G$ using $L(\cdot)$ is not possible. 
\end{proof}

\subsection{Palette-Sparsification for Smaller Number of Available Colors}\label{app:sparsification-general}

In light of our palette-sparsification result for $(\Delta+1)$ coloring in Theorem~\ref{thm:color-sampling}, it is natural to wonder if 
a qualitatively similar statement also holds for the $c$-coloring problem for values of $c$ strictly smaller than $(\Delta+1)$.  
The following basic proposition suggests that a ``useful analogue'' of our palette-sparsification (in the context of sublinear algorithms and sparsifying the input graph) does not exist for $c < \Delta+1$ colors. 

\begin{proposition}\label{prop:sparsification-general}
	Let $k \geq 0$ be an integer. There exists some $(\Delta+1 - k)$-colorable graph $G(V,E)$ with maximum degree at most $\Delta$ such that 
	if for every vertex $v \in V$, we independently sample a set $L(v)$ of colors of size $o(\Delta^{1-1/(k+1)})$ uniformly at random from $[\Delta+1-k]$, then with probability $1-o(1)$, 
	there exists \emph{no} proper coloring function $\col: V \rightarrow [\Delta+1-k]$ of $G$ such that for all vertices $v \in V$, $\col(v) \in L(v)$. 
\end{proposition}
\begin{proof}
	Let $G(V,E)$ be the graph on $(\Delta+1)$ vertices obtained from a $(\Delta+1)$-clique by removing all edges between $k+1$ arbitrary vertices $v_1,\ldots,v_{k+1}$. Clearly, $G$ is $(\Delta+1-k)$-colorable since we can color
	$v_1,\ldots,v_{k+1}$ all with the same color, and use the remaining $\Delta-k$ colors to color the remaining $(\Delta+1-(k+1)) = \Delta-k$ vertices with unique colors.   
	However, note that for us to be able to list-color $G$ with colors from lists $L(\cdot)$, at least one color $c \in [\Delta+1-k]$ should appear in $L(v_1) \cap L(v_2) \cap \ldots \cap L(v_{k+1})$. A simple calculation implies that probability of 
	this event is only $o(1)$ when each vertex samples only $o(\Delta^{1-1/(k+1)})$ colors, finalizing the proof. 
\end{proof}

Proposition~\ref{prop:sparsification-general} implies that already for an analogue of palette-sparsification for $\Delta$-coloring, we need lists of size $\Omega(\sqrt{\Delta})$ size and the lists become even larger once we move
to $c$-coloring for $c < \Delta$. On the other hand, even for $\Delta$-coloring, if we end up sampling $\Omega(\sqrt{\Delta})$ colors per vertex, then every edge $(u,v)$ of the graph has a constant probability of having $L(u) \cap L(v) \neq \emptyset$. 
This in turns implies that in terms of ``sparsification'', we can only hope to remove a constant fraction of the edges of the input graph in \emph{the best case scenario} (as we are \emph{not} claiming a matching upper
bound for Proposition~\ref{prop:sparsification-general}). This is in sharp contrast to our results for $(\Delta+1)$-coloring that allows
 us sparsify a graph with $\Omega(n^2)$ edges down to only $\Ot(n)$ edges.

\end{document}